\documentclass[12pt]{article}

\usepackage{amsmath, amssymb, amsthm}
\usepackage{mathrsfs}
\usepackage{bm}

\usepackage[utf8]{inputenc}
\usepackage[T1]{fontenc}
\usepackage{lmodern}
\usepackage{microtype}
\usepackage{setspace}
\usepackage{hyperref}
\usepackage{epigraph}
\usepackage{algorithm}
\usepackage{algpseudocode}
\usepackage{graphicx}
\usepackage{subcaption}

\newtheorem{theorem}{Theorem}[section]
\newtheorem{lemma}[theorem]{Lemma}

\usepackage[
  backend=biber,
  style=phys,
  sorting=none,
  biblabel=brackets
]{biblatex}
\title{Loss--Complexity Landscape \\ and \\ Model Structure Functions}
\author{Alexander Kolpakov\thanks{Corresponding author; University of Austin, 522 Congress Ave, Austin, TX 78701, USA; email: kolpakov.alexander@gmail.com}}
\date{\today}

\begin{document}
\maketitle

\begin{abstract}
We develop a framework for dualizing the Kolmogorov structure function $h_x(\alpha)$, which then allows using computable complexity proxies. We establish a mathematical analogy between information-theoretic constructs and statistical mechanics, introducing a suitable partition function and free energy functional. We explicitly prove the Legendre–Fenchel duality between the structure function and free energy, showing detailed balance of the Metropolis kernel, and interpret acceptance probabilities as information-theoretic scattering amplitudes. A susceptibility-like variance of model complexity is shown to peak precisely at loss-complexity trade-offs interpreted as phase transitions. Practical experiments with linear and tree-based regression models verify these theoretical predictions, explicitly demonstrating the interplay between the model complexity, generalization, and overfitting threshold.
\end{abstract}

\section{Introduction}
Kolmogorov complexity $K(x)$ of a string $x$ is the length of the shortest program that outputs $x$ on a fixed universal prefix-free Turing machine, and plays a fundamental role in theoretical computer science and information theory \cite{li1997introduction}. Its refinement, the Kolmogorov structure function, measures the complexity required to describe a string within a given range of model set complexity:
\[
h_x(\alpha)=\min_{\substack{S\ni x \\ K(S)\le\alpha}}\log|S|,
\]
where $K(S)$ is the Kolmogorov complexity of the set of binary strings that provide models for the given binary string $x$. 

Here the minimization is performed, for each fixed complexity budget $\alpha$, over all model classes $S$ containing $x$ whose complexity does not exceed $\alpha$; in this sense, $\alpha$ parametrizes the admissible region in complexity space. The quantity $\log|S|$ is the remaining description length once the model class $S$ has been specified. A detailed account is given in~\cite{vereshchagin00} which is accessible exclusively to readers with specialized domain knowledge. Thus we give a self-contained account that only assumes a solid general mathematical culture by simply avoid Kolmogorov complexity while scaffolding a similar idea.  

Notably, practical applications seem impossible due to the uncom\-pu\-ta\-bi\-li\-ty of Kolmogorov complexity. Numerical approximations to $K(S)$ have been developed and studied \cite{Zenil2011}, though their efficiency is likely much lower than what is necessary for practical computations. Thus, it makes sense to replace $K(S)$ by a computable complexity proxy $\mathrm{Comp}(S)$, which gives rise to a variety of ``structure functions'' $h_x(\alpha)$ that can be investigated numerically, as described in Section~\ref{section:model}. 

This is not the same as working with Kolmogorov complexity itself, nor is it the classical minimum description length (MDL) framework: the latter is a coding-based model-selection principle tied to a chosen statistical encoding scheme, whereas here we retain only the structure-function viewpoint and replace $K(S)$ by an explicit computable proxy.

First, we study the corresponding free-energy functional dual to the structure function: the Legendre--Fenchel duality provides the equivalence between the constrained problem at fixed complexity budget $\alpha$ and the penalized problem with Lagrange multiplier $\lambda$, as made precise in Sections~\ref{section:action}--\ref{section:lf}. Then, we develop a rigorous method, based on simulated annealing, that elucidates explicit connections to statistical mechanics and scattering theory, cf. Sections~\ref{section:sm}--\ref{section:tradeoff}.

In practice, however, Bayesian optimization gives a more powerful approach to model optimization \cite{robson2013diffusion, seckler2022bayesian, reich2015probabilistic}. Coupled with our method, it provides a way of model optimization that finds a model with good generalization properties and avoids overfitting, cf. Section~\ref{section:experiments}.   

\section{Model Structure Function}\label{section:model}

Let $S$ be our model viewed as a set of trainable parameters and hyperparameters. We shall write $S \ni x$ (read ``$S$ models $x$'') if our model has $x$ in its training set. The test (and validation, if any) set of $S$ is not considered until later. 

A computable complexity function $\mathrm{Comp}(S)$ together with the \textit{training} loss function $\mathrm{Loss}(S)$ leads to a new \emph{model structure function} of the form
\[
h_x(\alpha)=\min_{\substack{S\ni x\\\mathrm{Comp}(S)\le\alpha}}\mathrm{Loss}(S).
\]
This function reflects the tradeoff between the model complexity and bias. One may think of it as a quantification of the bias--variance tradeoff in terms of model's complexity rather than its variance. Indeed, high-variance models are necessarily more complex, but complex models are not necessarily high-variance \cite{neal2019modern}.  

We verify numerically and argue theoretically that we obtain the loss-complexity landscape for our model depending on its hyperparameters, and that the model structure function captures the practically significant aspects of it, such as the complexity salience point after which overfitting occurs. 

\section{Information-theoretic Action and Free Energy}\label{section:action}
Let us define an information-theoretic action functional (hereafter called ``action'' for brevity) as
\[
A_\lambda(S)=\lambda\,\mathrm{Comp}(S)+\mathrm{Loss}(S).
\]

This action serves as a computational analogue of physical action in statistical physics. Minimizing this yields a free-energy analogue
\[
F(\lambda)=\min_{S\ni x} A_\lambda(S),
\]
where we want to minimize both the model's complexity and its training loss, provided a given balance between the two parts described by $\lambda \geq 0$. That is, if $\lambda = 0$ then we minimize the training loss at the possible cost of high complexity, and if $\lambda \gg 1$ then we try to learn the train dataset with a parsimonious model. 

\section{Legendre–Fenchel Duality Between $h_x(\alpha)$ and $F(\lambda)$}\label{section:lf}

The reason behind the free energy functional becomes apparent in the context of duality. Let $\aleph$ be the set of all possible model complexity values. Extend $h_x(\alpha)$ to a convex function \(\phi\colon\mathbb R\to\mathbb R\cup\{+\infty\}\) by setting
\[
\phi(\alpha)=
\begin{cases}
h_x(\alpha),&\alpha\ge 0,\;\alpha\in\aleph,\\
+\infty,&\text{otherwise.}
\end{cases}
\]

\begin{theorem}[Legendre–Fenchel Duality]
The functions \(\phi(\alpha)\) and \(F(\lambda)\) are Legendre–Fenchel duals:
\[
F(\lambda)
\;=\;
\min_{\alpha\ge0}\bigl[\lambda\,\alpha\;+\;\phi(\alpha)\bigr],
\qquad
\phi(\alpha)
\;=\;
\max_{\lambda\ge0}\bigl[F(\lambda)\;-\;\lambda\,\alpha\bigr].
\]
\end{theorem}

\begin{proof}
Observe that
\[
\min_{\alpha\ge0}\bigl[\lambda\alpha + \phi(\alpha)\bigr]
=\min_{\alpha\,\in\,\aleph}\bigl[\lambda\,\alpha + h_x(\alpha)\bigr]
=\min_{S\ni x}\bigl[\lambda\,\mathrm{Comp}(S)+\mathrm{Loss}(S)\bigr]
=F(\lambda),
\]
since for each \(S\) with \(\mathrm{Comp}(S)=\alpha\), \(\mathrm{Loss}(S)\ge h_x(\alpha)\) and equality is attained by the definition of \(h_x(\alpha)\).  This establishes the first identity.

By definition of the model structure function,
\[
h_x(\alpha)
=\min_{\substack{S\ni x \\ \mathrm{Comp}(S)\le\alpha}}
\mathrm{Loss}(S).
\]
Introduce a nonnegative Lagrange multiplier \(\lambda\ge0\) to enforce the constraint \(\mathrm{Comp}(S)\le\alpha\), and  let
\[
L(S,\lambda)
=\mathrm{Loss}(S)\;+\;\lambda\bigl(\mathrm{Comp}(S)-\alpha\bigr)
\]
be the associated Lagrangian. 

For any model \(S\) satisfying \(\mathrm{Comp}(S)\le\alpha\), we have
\(\mathrm{Loss}(S) \geq L(S,\lambda)\).  Therefore
\[
h_{x}(\alpha)
=\min_{S\ni x}\,\max_{\lambda\ge0}L(S,\lambda).
\]
Since \(L(S,\lambda)\) is linear (and thus convex) in \(\lambda\) and the set of models is countable\footnote{Either finite or can be enumerated with natural numbers, e.g. we can enumerate all Turing machines by the so-called G\"odel numbering or, for all practical purposes, we have a finite number of real numbers available for a computer, depending on the bit size, and a potentially infinite but discreet set of neural network architectures, as they are easily described as graphs.}, we can swap minimization and maximization:
\[
h_x(\alpha)
=\max_{\lambda\ge0}\;\min_{S\ni x}L(S,\lambda)
=\max_{\lambda\ge0}\;\min_{S\ni x}\bigl[\mathrm{Loss}(S)+\lambda\,\mathrm{Comp}(S)-\lambda\,\alpha\bigr].
\]
By definition of the free energy, we have
\[
\min_{S\ni x}\bigl[\mathrm{Loss}(S)+\lambda\,\mathrm{Comp}(S)\bigr]
=F(\lambda).
\]
Hence
\[
h_x(\alpha)
=\max_{\lambda\ge0}\bigl[F(\lambda)-\lambda\,\alpha\bigr],
\]
as required.
\end{proof}

Let us note that the free energy
\[
F(\lambda)\;=\;\min_{\alpha\ge0}\bigl[\lambda\,\alpha \;+\;h_x(\alpha)\bigr]
\]
can be viewed as the \emph{lower envelope} of the family of lines
\[\ell_\alpha(\lambda)=\lambda\,\alpha + h_x(\alpha),\] 
one for each complexity level \(\alpha\).  

Geometrically, \(F\) is the pointwise infimum of these lines, and is therefore a convex, piecewise‐linear function of \(\lambda\).  On each interval where a single line \(\ell_{\alpha_*}\) attains the minimum, we have \(F(\lambda)=\ell_{\alpha_*}(\lambda)\) and its slope is constant, equal to \(\alpha_*\), the model complexity.  When two lines \(\ell_{\alpha_1}\) and \(\ell_{\alpha_2}\) intersect, the minimizer switches from \(\alpha_1\) to \(\alpha_2\), producing a ``kink'' or ``elbow'' shape in the graph of \(F\).

That’s why, even if $h_x(\alpha)$ is complicated, the dual free energy form $F(\lambda)$ automatically acquires a piecewise‐linear convex envelope structure: this shape is much easier to analyze, especially if we are interested only in the extrema of $F(\lambda)$. 

\section{Statistical Mechanics Analogy}\label{section:sm}
Let us introduce the following partition function
\[
Z(\lambda,T)=\sum_{S\ni x} e^{-A_\lambda(S)/T},
\]
defining a Gibbs probability distribution over model classes as
\[
\pi_\lambda(S)=\frac{e^{-A_\lambda(S)/T}}{Z(\lambda,T)}.
\]

Then the free energy can defined analogously to statistical mechanics as:
\[
F(\lambda,T)=-T\log Z(\lambda,T).
\]

In the low-temperature limit $T\to 0$, we recover precisely $F(\lambda)$. Otherwise, we can think about ``least action'' models being more likely with respect to the Gibbs measure. In this setting, running a probabilistic search algorithm with respect to $\pi_\lambda(S)$ should reveal a model with relatively low $A_\lambda(S)$, which is advantegeous for us. 

\section{Metropolis--Hastings Algorithm}\label{section:metropolis}

The Metropolis algorithm, introduced by Metropolis \emph{et al.}~\cite{metropolis53}, is an important method in computational statistical physics and optimization. It turns out to be particularly useful for sampling from complicated distributions and minimizing complex cost functions.

\subsection{Formal Definition}

Consider a finite or countable state space \(\mathcal{S}\) and a real-valued function \(A:\mathcal{S}\to\mathbb{R}\), referred to here as the action. The Metropolis algorithm constructs a Markov chain with stationary probability distribution:
\[
\pi(S)=\frac{e^{-A(S)/T}}{Z(T)},\quad\text{with}\quad Z(T)=\sum_{S'\in\mathcal{S}} e^{-A(S')/T},
\]
which is the Gibbs distribution discussed above, safe for a more general form for the action functional. 

\subsection{Algorithmic Procedure}

The Metropolis procedure for a single iteration is as follows.

\begin{algorithm}[H]
\caption{Metropolis Step}
\begin{algorithmic}[1]
\State Initialize current state \(S\in\mathcal{S}\).
\State Generate candidate state \(S'\) from a symmetric proposal distribution \(Q(S\to S')=Q(S'\to S)\).
\State Compute the action difference: \(\Delta A = A(S') - A(S)\).
\State Accept the new state \(S'\) with probability:
\[
P_{\text{accept}}(S\to S')=\min\{1, e^{-\Delta A/T}\}.
\]
\State Otherwise, retain the current state \(S\).
\end{algorithmic}
\end{algorithm}

Repeatedly applying this step produces a Markov chain whose stationary distribution is guaranteed to be the Gibbs distribution \(\pi(S)\) under some mild conditions on the Gibbs measure $\pi(S)$, and the proposal distribution $Q(S \to S')$. 

\subsection{Conditions for Correctness}

The following \emph{detailed balance} condition is sufficient for the existence of stationary distribution:
\[
\pi(S)P(S\to S') = \pi(S')P(S'\to S),
\]
where \(P(S\to S')\) denotes the transition probability from state \(S\) to \(S'\). We can readily show that the Gibbs measure defined above satisfies it. 

\begin{lemma}[Detailed Balance for Metropolis Algorithm]
The Metropolis transition probability \(P(S\to S')\) satisfies the detailed balance condition with respect to \(\pi(S)\).
\end{lemma}
\begin{proof}
Consider two states \(S, S'\in\mathcal{S}\). If \(\Delta A = A(S') - A(S)\leq 0\), then:
\[
P(S\to S')=1,\quad P(S'\to S)=e^{-\frac{A(S)-A(S')}{T}}.
\]
Thus:
\[
\pi(S)P(S\to S')=\frac{e^{-A(S)/T}}{Z},
\]
\[
\pi(S')P(S'\to S)=\frac{e^{-A(S')/T} e^{-(A(S)-A(S'))/T}}{Z}=\frac{e^{-A(S)/T}}{Z}.
\]

If \(\Delta A>0\), the roles reverse and a similar argument holds. Hence, detailed balance is satisfied.
\end{proof}

Another, necessary condition, is that the Markov chain thus obtained is $\pi(S)$--irreducible and aperiodic. This can be easily guaranteed by the appropriate choice of the proposal distribution $Q(S \to S')$ as described in the standard references \cite{metropolis53, kirkpatrick83}. 

\subsection{Temperature Parameter and Annealing Schedule}

The temperature \(T\) regulates the balance between exploration (accepting higher-action states to escape local minima) and exploitation (preferring lower-action states). Typically, simulated annealing involves systematically lowering \(T\) from an initial high temperature \(T_0\) to near-zero values:
\[
T_{k+1}=\gamma T_k,\quad 0<\gamma<1.
\]

This annealing schedule enables the algorithm to probabilistically converge toward global minima of the action as \(T\to 0\).

\section{Simulated Annealing Procedure}\label{section:anneal}
To practically minimize $A_\lambda(S)$, one can use simulated annealing with Metropolis updates, as described in \cite{kirkpatrick83}. Repeating the annealing procedure across a range of $\lambda$ values yields approximate structure function pairs $(\alpha,h)$.

\begin{algorithm}[H]
\caption{Simulated Annealing}
\begin{algorithmic}[1]
\State Initialize model $S\in\mathcal{S}$ arbitrarily, set $T=T_0$
\While{$T>T_{\min}$}
\State Propose new model $S'$ from the neighborhood of current $S$
\State Compute action difference $\Delta A=A_\lambda(S')-A_\lambda(S)$
\State Accept $S'$ with probability $\min\{1,e^{-\Delta A/T}\}$
\State Decrease temperature $T\leftarrow\gamma T$ for $0<\gamma<1$
\EndWhile
\State Return approximate minimizer $S^*\approx S$
\end{algorithmic}
\end{algorithm}

Practically though, we would use Bayesian optimizers such as HyperOpt \cite{bergstra2013making} or Optuna \cite{akiba2019optuna}. 

\section{Information–Scattering Analogy}\label{section:scatter}
We already know that the detailed balance condition is satisfied with respect to the Gibbs measure $\pi_\lambda(S)$. This statistical mechanics analogy can be taken further with the following observation. However, we shall use it mostly as a useful analogy rather than an actual technique. 

\begin{theorem}[Acceptance as Scattering Amplitude]
The Metropolis acceptance criterion directly corresponds to a semiclassical path-integral scattering amplitude, identifying $T$ with Planck's constant $\hbar$:
\[
P(S\to S')\sim e^{-\Delta A/T}\leftrightarrow e^{i A[q(t)]/\hbar}.
\]
\end{theorem}
\begin{proof}
In simulated annealing, the Metropolis--Hastings acceptance probability for transitioning from a current state $S$ to a candidate state $S'$ is given explicitly by~\cite{kirkpatrick83}:
\[
P(S \to S') = \min\left\{1, e^{-(A(S') - A(S))/T}\right\},
\]
where $A(S)$ is the defined information-theoretic action and $T$ is the annealing temperature.

To reveal the analogy with quantum-mechanical scattering amplitudes, consider Feynman's path-integral formulation of quantum mechanics. The quantum mechanical amplitude for a system transitioning from state $q$ at time $0$ to state $q'$ at time $t$ is given by the path integral~\cite{feynman2010quantum}:
\[
\langle q'| e^{-iHt/\hbar} |q\rangle = \int \mathcal{D}[q(t)]\, e^{iA[q(t)]/\hbar},
\]
where $A[q(t)]$ is the classical action functional\footnote{In physics, it would be denoted $S[q(t)]$, but letter ``$S$'' is already used in another context.}, $H$ is the Hamiltonian of the system, and $\hbar$ is Planck's constant.

In the semiclassical or stationary-phase approximation, the path integral is dominated by paths close to classical solutions~\cite{schulman2005techniques}. Expanding around these solutions, one obtains an amplitude dominated by terms of the form:
\[
e^{i A[q_{\text{cl}}(t)]/\hbar}.
\]

Specifically, when considering quantum tunneling through potential barriers (classically forbidden regions), the transition amplitude takes a form proportional to~\cite{landau2013quantum}:
\[
e^{-(A_{\text{barrier}} - A_{\text{initial}})/\hbar}.
\]

Thus, if we identify the temperature parameter $T$ of simulated annealing with Planck's constant $\hbar$,
\[
T \longleftrightarrow \hbar,
\]
and the information-theoretic action difference $\Delta A = A(S') - A(S)$ with the corresponding classical action difference $A_{\text{barrier}} - A_{\text{initial}}$, the following formal analogy appears: 
\[
e^{-(A(S') - A(S))/T} \quad\leftrightarrow\quad e^{-(A_{\text{barrier}} - A_{\text{initial}})/\hbar}.
\]

Hence, the Metropolis acceptance criterion explicitly mirrors semiclassical quantum tunneling amplitudes, in analogy between simulated annealing acceptance probabilities and quantum-mechanical scattering amplitudes derived from the stationary-phase approximation of path integrals.
\end{proof}

\section{Susceptibility and Resonance as Trade-off between Loss and Complexity}\label{section:tradeoff}

In statistical‐mechanical language the \emph{susceptibility} measures the sensitivity of the free energy to changes in the Lagrange multiplier \(\lambda\).  Concretely, let the partition function be
\[
Z(\lambda)=\sum_{S\ni x}e^{-A_\lambda(S)},
\qquad
A_\lambda(S)=\lambda\,\mathrm{Comp}(S)+\mathrm{Loss}(S),
\]
so that the free energy is
\[
F(\lambda)=-\ln Z(\lambda).
\]

By standard thermodynamic identities,
\[
\frac{dF}{d\lambda}
=\bigl\langle\mathrm{Comp}(S)\bigr\rangle_\lambda,
\quad
\frac{d^2F}{d\lambda^2}
=\frac{d}{d\lambda}\bigl\langle\mathrm{Comp}(S)\bigr\rangle_\lambda
=\mathrm{Var}_{\pi_\lambda}\!\bigl[\mathrm{Comp}(S)\bigr],
\]
where \(\langle\cdot\rangle_\lambda\) denotes expectation under the Gibbs measure
\[\pi_\lambda(S)=e^{-A_\lambda(S)} / Z(\lambda).\]  

Therefore, we set
\[
\chi(\lambda)=\frac{d^2F}{d\lambda^2}
=\mathrm{Var}_{\pi_\lambda}\!\bigl[\mathrm{Comp}(S)\bigr].
\]

\subsection{Two competing Models}
Intuitively, \(\chi(\lambda)\) quantifies how many different models \(S\) of varying complexity contribute to the free energy at a given \(\lambda\).  A large \(\chi\) means the Gibbs weight is split between two (or more) widely differing complexity levels, signaling a \emph{phase transition} in the loss-complexity landscape.

\begin{theorem}[Susceptibility Resonance]
Let \(S_1\) and \(S_2\) be the two lowest‐action configurations at a given \(\lambda\), with
\[
A_i = A_\lambda(S_i),
\quad
C_i = \mathrm{Comp}(S_i),
\quad
i=1,2,
\]
and assume all other \(S\) have strictly larger action.  Then
\[
\arg\max_{\lambda}\chi(\lambda)
\;=\;\{\lambda : A_1(\lambda)=A_2(\lambda)\},
\]
i.e.\ \(\chi\) is maximized exactly when the two different model's actions coincide. 
\end{theorem}

\begin{proof}
In the two‐state approximation the partition function is
\[
Z \approx e^{-A_1} + e^{-A_2},
\]
and the Gibbs probabilities are
\[
\pi_i \;=\;\frac{e^{-A_i}}{e^{-A_1}+e^{-A_2}},
\quad i=1,2.
\]
The complexity variance reduces to
\[
\chi
=\pi_1\,\pi_2\,(C_1 - C_2)^2.
\]
Since \(C_1\neq C_2\) by assumption, the factor \((C_1-C_2)^2\) is constant in \(\lambda\), and
\[
\pi_1\,\pi_2
=\frac{e^{-A_1}e^{-A_2}}{(e^{-A_1}+e^{-A_2})^2}
\]
is maximized precisely when \(e^{-A_1}=e^{-A_2}\), i.e.\ \(A_1=A_2\).  Hence \(\chi(\lambda)\) peaks exactly at a \emph{resonance} point.
\end{proof}

Thus, assume that we have exactly two candidate models \(S_1,S_2\) with complexities \(C_i=\mathrm{Comp}(S_i)\) and losses \(L_i=\mathrm{Loss}(S_i)\). Each model's information-theoretical action is
\[
A_\lambda(S_i)=L_i+\lambda\,C_i,\quad i=1,2.
\]

A {\em resonance} at \(\lambda^*>0\) means
\[
L_1 + \lambda^* C_1 \;=\; L_2 + \lambda^* C_2,
\]
which is equivalent to 
\[
\lambda^* = \frac{L_1 - L_2}{C_2 - C_1}.
\]
Since by assumption \(C_1\neq C_2\), the losses of these two models must differ exactly by \(\lambda^*(C_2 - C_1)\), where $\lambda^*$ measures how strongly the difference in complexity affects the goodness-of-fit.  

Moreover,  \(\lambda^*>0\) occurs if and only if
\[
\frac{L_1 - L_2}{C_2 - C_1} > 0
\quad\Longleftrightarrow\quad
(L_1 - L_2)\;(C_2 - C_1) > 0.
\]
Thus, if \(C_2>C_1\) (i.e.\ \(S_2\) is more complex) then \(L_2<L_1\): the more complex model must also fit strictly better (smaller loss) for a crossing at positive \(\lambda\). If \((L_1 - L_2)(C_2 - C_1)\le0\), the two lines
\(\ell_i(\lambda)=L_i+\lambda C_i\) never meet for \(\lambda>0\).  One model then \emph{dominates} the envelope for all \(\lambda\), and there is no resonance.

As mentioned in the previous discussion of duality, plotting the lines \(\ell_i(\lambda) = L_i + \lambda C_i\), $i=1,2$, as functions of \(\lambda\), a positive‐\(\lambda\) intersection at \(\lambda^*\) produces exactly the ``kink'' or ``elbow'' in the lower envelope
\[
F(\lambda)=\min\{\ell_1(\lambda),\ell_2(\lambda)\}.
\]
For \(\lambda<\lambda^*\), the line with smaller loss \(L_i\) (but higher complexity $C_i$) attains the minimum; for \(\lambda>\lambda^*\), the line with smaller complexity \(C_i\) (but large loss $L_i$) takes over.  The switch at \(\lambda^*\) yields a slope discontinuity \(C_1\to C_2\), and hence a peak in the susceptibility \(\chi=d^2F/d\lambda^2\).
 
Thus, a positive resonance \(\lambda^*>0\) signals a genuine trade‐off between model's fit and complexity: the more complex model must achieve lower loss to ever be preferred.  The location of \(\lambda^*\) quantifies the exact balance point.  If no such positive resonance exists, one model is uniformly better (either strictly simpler with no loss penalty, or strictly better‐fitting with no complexity penalty), and no ``elbow'' appears in the plot of \(F(\lambda)\).

\subsection{General \(k\)-state Resonance}

Suppose that at a critical \(\lambda^*\) exactly \(k\) models \(S_1,\dots,S_k\) share the minimal action
\[
A_i^* \;=\; A_{\lambda^*}(S_i)\quad(i=1,\dots,k),
\]
and all other models have strictly larger action.  Write their complexities as \(C_i=\mathrm{Comp}(S_i)\).  Near \(\lambda^*\), let
\[
\lambda = \lambda^* + \varepsilon,
\qquad
A_i(\lambda)=A_i^*+\varepsilon\,C_i,
\]
and work in the two‐term low‐temperature (or \(T=1\)) Gibbs approximation
\[
Z(\varepsilon)\approx\sum_{i=1}^k e^{-A_i(\lambda)}
=\sum_{i=1}^k e^{-A_i^*}\,e^{-\varepsilon C_i}
= e^{-A^*}\sum_{i=1}^k e^{-\varepsilon C_i},
\]
where \(A^*=A_i^*\) for the degenerate minima.  The Gibbs weights become
\[
P_i(\varepsilon)
=\frac{e^{-\varepsilon C_i}}{\sum_{j=1}^k e^{-\varepsilon C_j}}.
\]
The susceptibility is
\[
\chi(\varepsilon)
=\mathrm{Var}_{P(\varepsilon)}[C]
=\sum_{i=1}^k P_i(\varepsilon)\,C_i^2
-\Bigl(\sum_{i=1}^k P_i(\varepsilon)\,C_i\Bigr)^2.
\]

\paragraph{Stationarity at \(\varepsilon=0\).}
At \(\varepsilon=0\), all \(P_i(0)=1/k\) and
\(\sum_iP_i'C_i=0\) by symmetry, so
\(\chi'(\varepsilon)\big|_{\varepsilon=0}=0\).  Thus \(\chi\) is stationary at \(\lambda^*\).

\paragraph{Second derivative and peak width.}
Compute the second derivative at \(\varepsilon=0\).  Expanding
\[
P_i(\varepsilon)
=\frac{1-\varepsilon C_i + O(\varepsilon^2)}{k - \varepsilon\sum_j C_j + O(\varepsilon^2)}
=\frac1k - \frac{\varepsilon}{k}\Bigl(C_i-\bar C\Bigr) + O(\varepsilon^2),
\quad
\bar C=\frac1k\sum_{j=1}^k C_j,
\]
one finds after straightforward algebra
\[
\chi(\varepsilon)
=\frac1k\sum_i(C_i-\bar C)^2
\;-\;\frac{\varepsilon^2}{k}\sum_i(C_i-\bar C)^4
\;+\;O(\varepsilon^3).
\]
The \(\varepsilon^2\) coefficient is strictly negative provided not all \((C_i-\bar C)\) vanish.  Hence \(\chi\) has a \emph{strict maximum} at \(\varepsilon=0\), i.e.\ at
\(\lambda=\lambda^*\).

\paragraph{Scaling of the peak width.}
The width \(\Delta\lambda\) over which \(\chi\) falls to half its peak value satisfies
\(\Delta\chi\approx -\,\tfrac12\chi''(0)\,(\Delta\lambda)^2\), so
\[
\Delta\lambda \;=\; O\!\left(\frac{\sum_i(C_i-\bar C)^2}{\sum_i(C_i-\bar C)^4}\right)^{\!1/2}
=\;O\!\bigl(\min_{i\ne j}|C_i-C_j|\bigr)^{-1}.
\]
Thus the resonance becomes sharper as the complexity‐gaps \(|C_i-C_j|\) grow, quantifying the universality of phase‐transition peaks in the loss-complexity landscape.

\subsection{Designing Complexity Functionals}
\label{subsec:comp-design}

The central choice in our framework is the computable complexity functional
\(\mathrm{Comp}(S)\).  In contrast to Kolmogorov complexity which is machine–dependent
but invariant up to additive constants (and fundamentally non-computable) we deliberately
work with task–dependent complexity proxies.  Each such choice
induces an \emph{a priori} different model structure function
\[
  h_x(\alpha)
  \;=\;
  \min_{\substack{S\ni x\\ \mathrm{Comp}(S)\le \alpha}}\mathrm{Loss}(S),
\]
and, via Legendre--Fenchel duality, different free energy \(F(\lambda)\) and
susceptibility \(\chi(\lambda)\).  The phase transitions we observe are therefore
always relative to an explicit notion of complexity chosen by the practitioner.

Thus, we only require the following mild axioms to hold for the complexity proxy
$\mathrm{Comp}$:
\begin{itemize}
  \item \emph{Monotonicity:} if a model class $S'$ strictly enlarges $S$ (e.g. by increasing depth, width, or number of parameters),
        then
        \[
          \mathrm{Comp}(S') \;\ge\; \mathrm{Comp}(S) - O(1),
        \]
        hence making the class more expressive cannot substantially
        decrease its complexity;

  \item \emph{Approximate additivity:} for composite models $S = S_1 \cup S_2$ such that $S_1 \cap S_2 = \emptyset$,
        \[
          \mathrm{Comp}(S)
          \;=\;
          \mathrm{Comp}(S_1) + \mathrm{Comp}(S_2) + O(1),
        \]
        where the $O(1)$ term does not depend on the particular models
        $S_1,S_2$ but only on the choice of encoding;

  \item \emph{Invariance under trivial reparametrizations:}
        if $S$ and $S'$ are different encodings of the same architecture
        and hyperparameters (e.g.\ reshuffled parameter vector, equivalent
        graph representation), then
        \[
          \mathrm{Comp}(S') \;=\; \mathrm{Comp}(S) + O(1).
        \]
\end{itemize}

When \(\mathrm{Comp}(S)\) is chosen as an explicit codelength for the model class
\(S\) (architecture, parameters and training scheme), the action
\[
  A_\lambda(S) \;=\; \mathrm{Loss}(S) + \lambda\,\mathrm{Comp}(S)
\]
is directly analogous to a \emph{minimum description length} (MDL) objective. The
data misfit \(\mathrm{Loss}\) plays the role of the negative log-likelihood and
\(\mathrm{Comp}(S)\) plays the role of the model codelength.  In classical MDL one
fixes a particular coding scheme and seeks the model minimizing the sum of these
two terms.  We make the dependence on the Lagrange multiplier
\(\lambda\) explicit, compute the dual free energy \(F(\lambda)\), and study its
derivatives (in particular, the susceptibility \(\chi(\lambda)\)) as witnesses
of phase transitions in the loss–complexity landscape.

Beyond MDL-style codelengths, there is a growing body of work on computable
approximations to Kolmogorov complexity \cite{zenil2018, zenil2023, johnston2022, liu2022, zhang2024}.  These can
all be used as alternative choices of \(\mathrm{Comp}(S)\) in our framework:
once a computable complexity proxy is fixed, the Legendre--Fenchel duality and susceptibility arguments
proceed in complete analogy.  

In this sense our results are
agnostic to the particular proxy used, while different choices of \(\mathrm{Comp}(S)\)
become practical in order to emphasize different inductive biases (e.g.\ sparsity, hierarchical structure, or compressibility) in the resulting phase–transition analysis.

\section{Numerical Validation and Experiments}\label{section:experiments}

\subsection{Experiment Setup}

In order to test the the above theory computationally, we implement several regression tasks: polynomial regression, Fourier expansions, and tree-based models. All experiments explicitly match theoretical predictions and align with empirical overfitting thresholds. 

The task at hand is to learn a function from noisy data. The function, for simplicity, being just $f(x) = \sin 2 n \pi x$, $x \in [0,1]$. The noise is sampled from the normal distribution $N(0, \sigma^2)$. 

The choice of the complexity function is very obvious: number of coefficients $\mathrm{Comp}(S) = d+1$ for polynomial degree $d$ regression, the number $\mathrm{Comp}(S) = 2 d + 1$ of Fourier coefficients up to mode $d$, and $\mathrm{Comp}(S) = d$ for the depth $d$ tree regressor. 

Thus, we may simply set 
\[
\mathrm{Comp}(S) = d, \; \mathrm{Loss}(S) = \mathrm{MSE}_{train}(S),
\]
in all three cases. Here, the mean squared error $\mathrm{MSE}_{train}$ is measured on the noisy \emph{train dataset}, while a clean and noisy \emph{test datasets} are kept apart. 

In our experiment we vary the Lagrange multiplier \(\lambda\) in  
\[
A_\lambda(d)\;=\;\mathrm{Loss}(d)\;+\;\lambda\,d,
\]  
and for each \(\lambda\) record the optimal depth  
\[
d^*(\lambda)\;=\;\arg\min_{d}\bigl[\mathrm{Loss}(d)+\lambda\,d\bigr].
\]  
Setting \(\alpha=d^*(\lambda)\) and  
\[
h(\alpha)\approx\mathrm{Loss}\bigl(d^*(\lambda)\bigr)
\]  
produces a discrete approximation to the structure function  
\[
h(\alpha)
=\min_{d\le\alpha}\mathrm{Loss}(d).
\]  
The pronounced ``elbow'' in the test‐MSE versus depth curve is exactly the dual reflection of the ``kink'' in the free energy  
\[
F(\lambda)=\min_d\bigl[\mathrm{Loss}(d)+\lambda\,d\bigr],
\]  
which occurs at the critical \(\lambda\) where two complexities exchange as the global minimizer.  By Legendre–Fenchel duality, both are manifestations of the same underlying phase‐transition phenomenon. 

After plotting the shape of $h(\alpha)$ with $\mathrm{Loss}$ being the train loss $\mathrm{MSE}_{train}$, we make an analogous plot for $\mathrm{MSE}_{test}$, on the test dataset (which can be either clean or noisy). The idea is to compare the shapes: for example, the train and test shapes may be similar, or the test shape may show overfitting for higher model complexities. 

\subsection{Linear models}

In this section, we perform two simple experiments on learning the noisy function $f(x) = \sin 2 n \pi x + \varepsilon$, $x \in [0,1]$, $n\in\{4,6\}$, $\varepsilon \sim N(0,\sigma^2)$ by way of linear regression. Ostensibly, we use polynomial and Fourier regressors, however the nature of such regressors is known to be linear (they are nothing more than projections on linear subspaces in function spaces). The Loss vs. Complexity curves are depicted in Figure~\ref{fig:poly} (polynomial regression) and Figure~\ref{fig:four} (Fourier series). 

The phase transitions between models delivering qualitatively different goodness-of-fit are visible in Figures~\ref{fig:poly-low-4pane}--\ref{fig:poly-high-4pane}.

Note that the phase transition happen at the same levels of complexity (compare the left and right panes of Figure~\ref{fig:poly}) independent of the noise level, and the qualitative behavior of the approximating polynomials also remains similar in Figure~\ref{fig:poly-low-4pane} and Figure~\ref{fig:poly-high-4pane}. Also, note that the loss on the \emph{clean} test dataset is always getting low as complexity grows independent on the level of noise: the latter, of course, is reasonably low in both cases, though differs by an order of magnitude. This shows that linear models are relatively robust to overfitting.

The picture is almost evident for the Fourier regression on $f(x) = \sin  4 \pi x + \varepsilon$, $x \in [0,1]$, $\varepsilon \sim N(0,\sigma^2)$, with both low $\sigma=0.05$ and high $\sigma=0.3$ noise. The complexity drops sharply as we reach the actual Fourier mode $\sin 2 \pi k x$ with $k=2$ leaving only the train and test (on the \emph{noisy} dataset) to differ, while the \emph{clean} test dataset loss confirms we have completely recovered the original generating model. This ``phase transition'' is indeed expected, which further confirms that our theory. In and of itself, however, this example may be less convincing: we approximate a trigonometric function via Fourier expansion, which is a mathematically trivial task.  

All code used to produce the above examples is available on GitHub, and the computation can be reproduced in the Google Colab environment \cite{github-structure}. 

\begin{figure}[!htbp]
    \centering
    \begin{subfigure}[t]{0.45\linewidth}
        \centering
        \includegraphics[width=\linewidth]{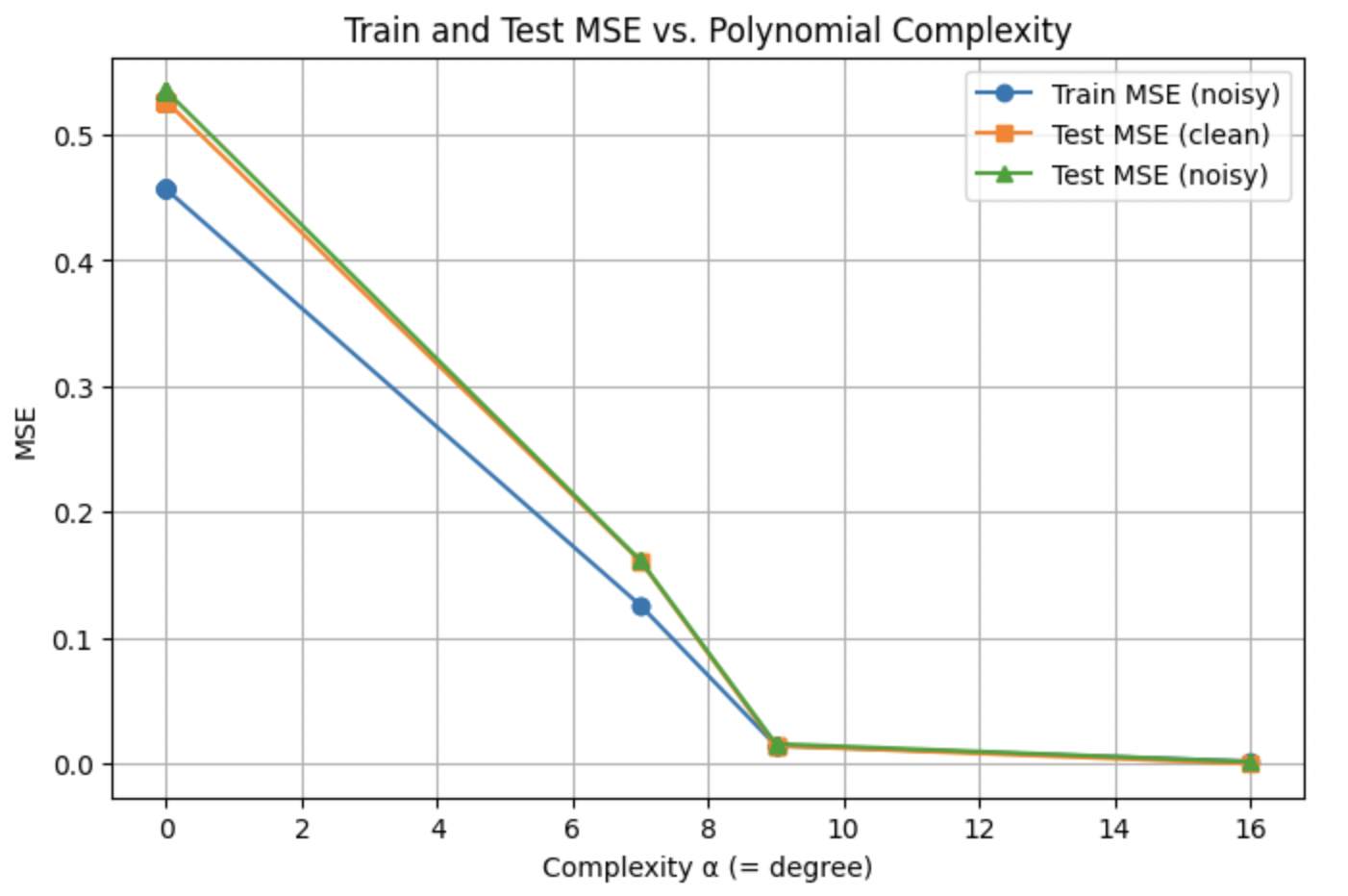}
        \caption{$\sigma = 0.05$}
        \label{fig:poly-low-noise}
    \end{subfigure}
    \hfill
    \begin{subfigure}[t]{0.45\linewidth}
        \centering
        \includegraphics[width=\linewidth]{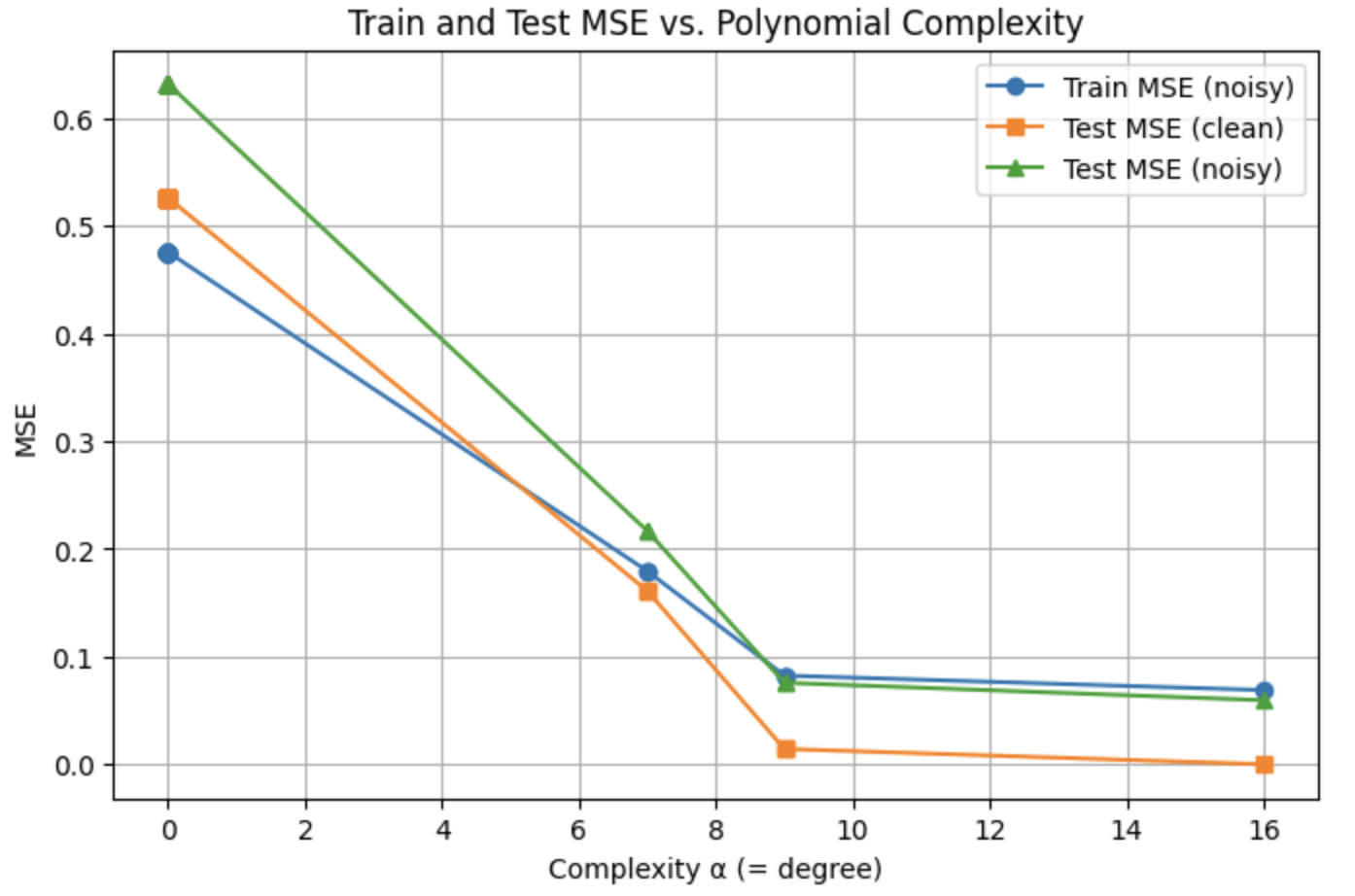}
        \caption{$\sigma=0.3$}
        \label{fig:poly-high-noise}
    \end{subfigure}
    \caption{
        Loss vs.\ Complexity for polynomial regression on 
        $f(x)=\sin(6\pi x) + \varepsilon$ with Gaussian noise $\varepsilon \sim N(0, \sigma^2)$
    }
    \label{fig:poly}
\end{figure}

\begin{figure}[!htbp]
  \centering
  %––––– Top‐left
  \begin{subfigure}[t]{0.48\textwidth}
    \centering
    \includegraphics[width=\linewidth]{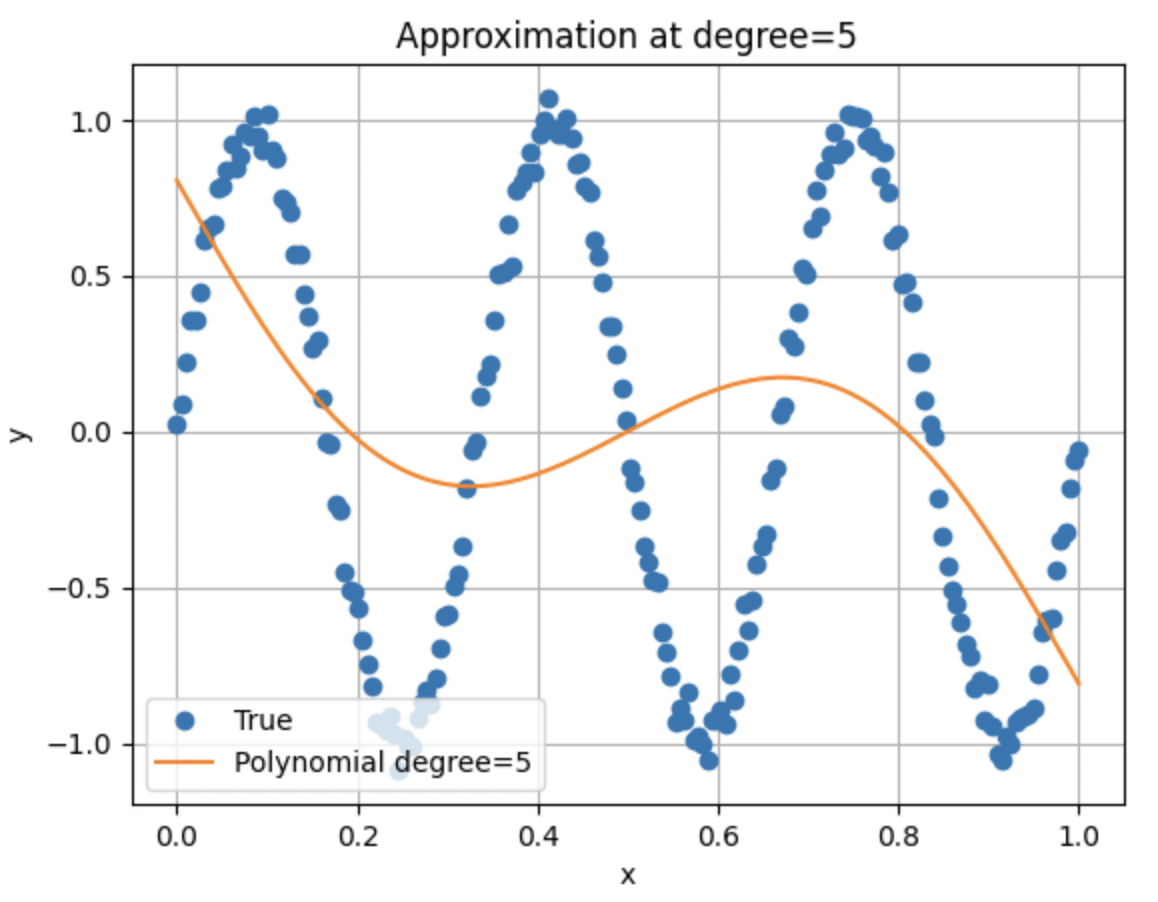}
    \caption{$d=5$, $\sigma=0.05$}
    \label{fig:poly-low-4pane-a}
  \end{subfigure}
  \hfill
  %––––– Top‐right
  \begin{subfigure}[t]{0.48\textwidth}
    \centering
    \includegraphics[width=\linewidth]{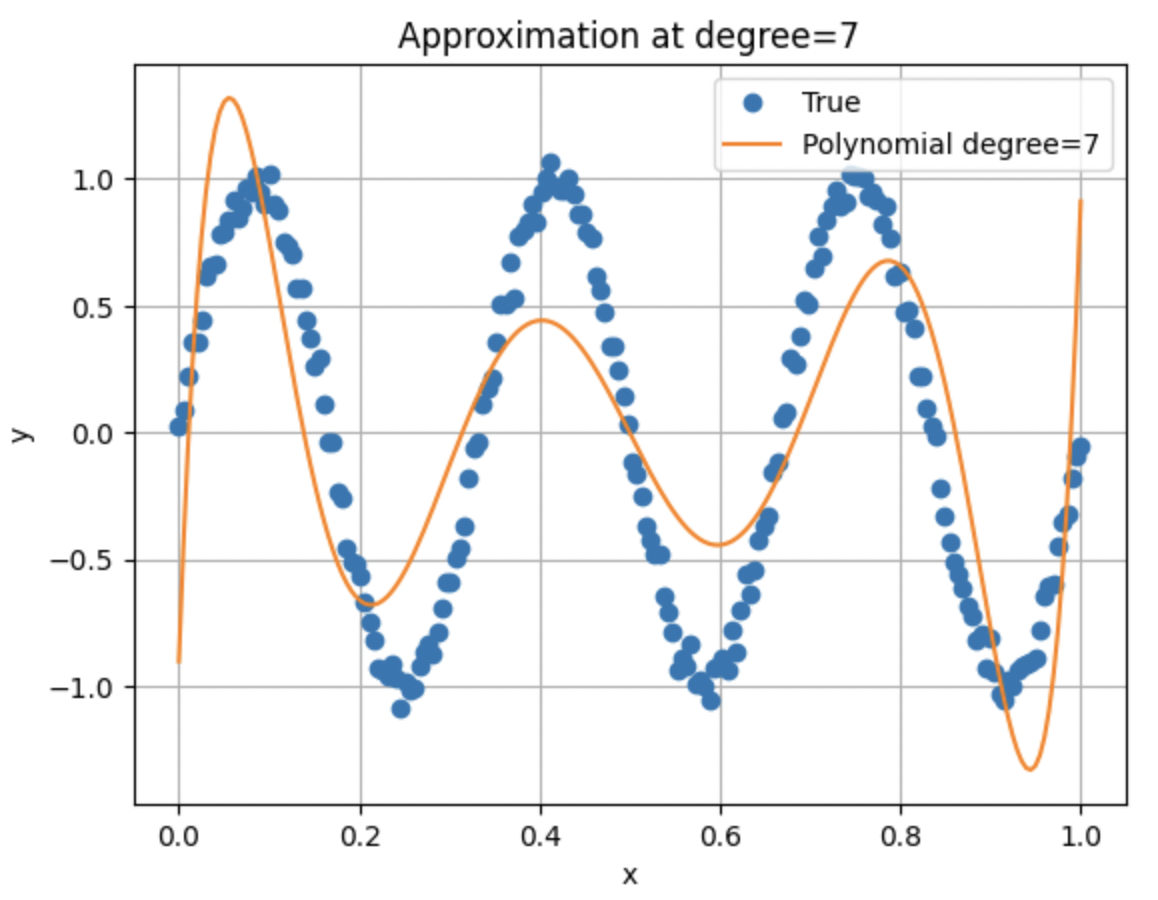}
    \caption{$d=7$, $\sigma=0.05$}
    \label{fig:poly-low-4pane-b}
  \end{subfigure}

  \vskip\baselineskip

  %––––– Bottom‐left
  \begin{subfigure}[t]{0.48\textwidth}
    \centering
    \includegraphics[width=\linewidth]{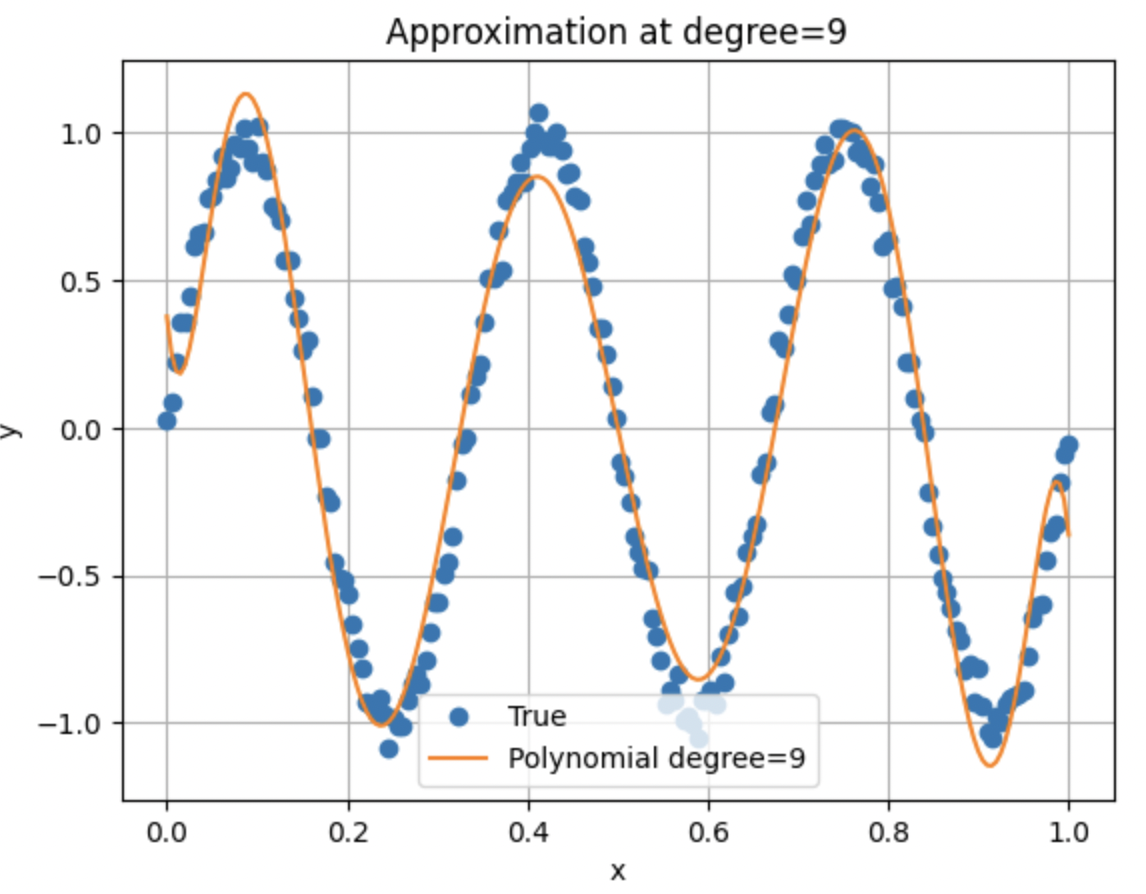}
    \caption{$d=9$, $\sigma=0.05$}
    \label{fig:poly-low-4pane-c}
  \end{subfigure}
  \hfill
  %––––– Bottom‐right
  \begin{subfigure}[t]{0.48\textwidth}
    \centering
    \includegraphics[width=\linewidth]{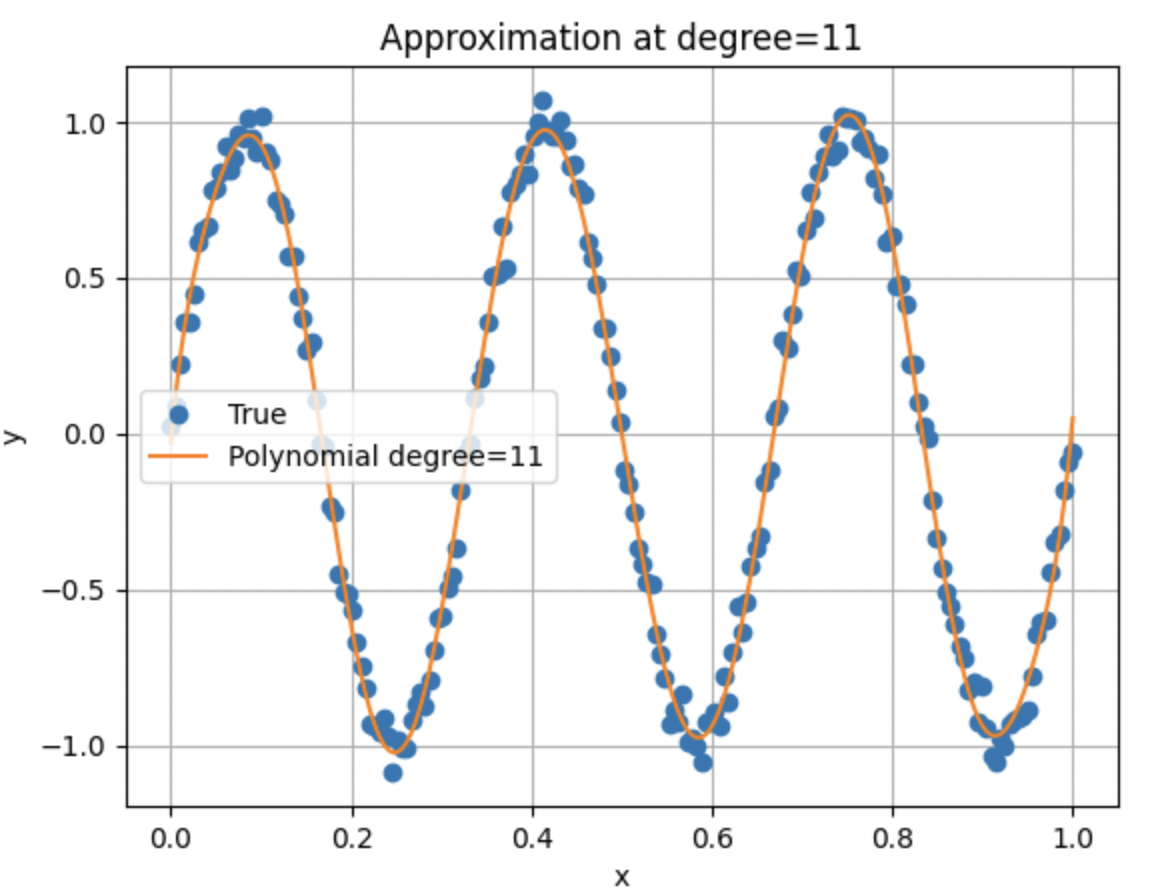}
    \caption{$d=11$, $\sigma=0.05$}
    \label{fig:poly-low-4pane-d}
  \end{subfigure}

  \caption{
  Polynomial regression for $f(x)=\sin(6\pi x) + \varepsilon$ with Gaussian noise $\varepsilon \sim N(0, \sigma^2)$ and varying polynomial degrees $d$. Note that the Loss vs. Complexity curve has ``elbows'' at $d=7$ and $d=9$. There are visible ``phase transitions'' in the shape of the polynomial vs the data  at $d = 5, 7, 9, 11$, while in between these values the regression curve shape stays relatively the same, and tends to stabilize after $d=11$. 
  }
  \label{fig:poly-low-4pane}
\end{figure}

\begin{figure}[!htbp]
  \centering
  %––––– Top‐left
  \begin{subfigure}[t]{0.48\textwidth}
    \centering
    \includegraphics[width=\linewidth]{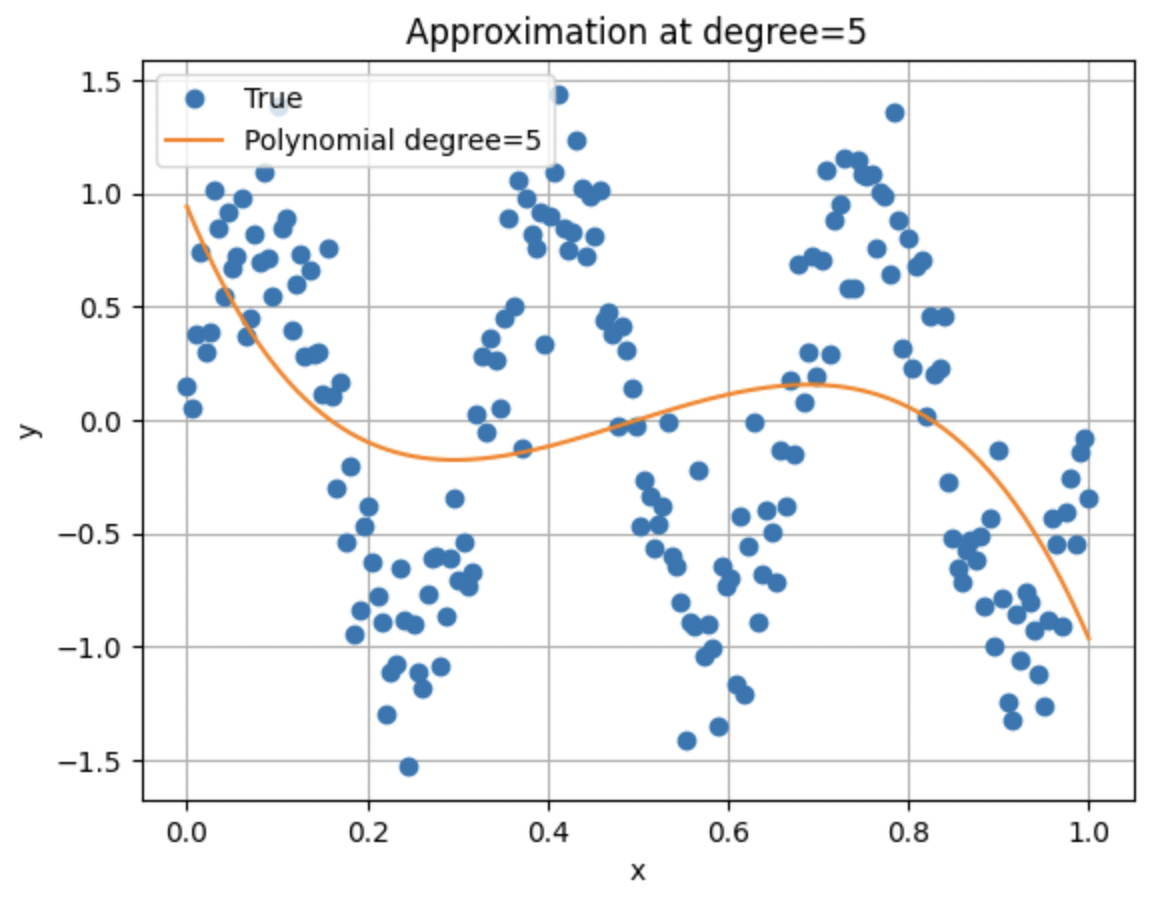}
    \caption{$d=5$, $\sigma=0.3$}
    \label{fig:poly-high-4pane-a}
  \end{subfigure}
  \hfill
  %––––– Top‐right
  \begin{subfigure}[t]{0.48\textwidth}
    \centering
    \includegraphics[width=\linewidth]{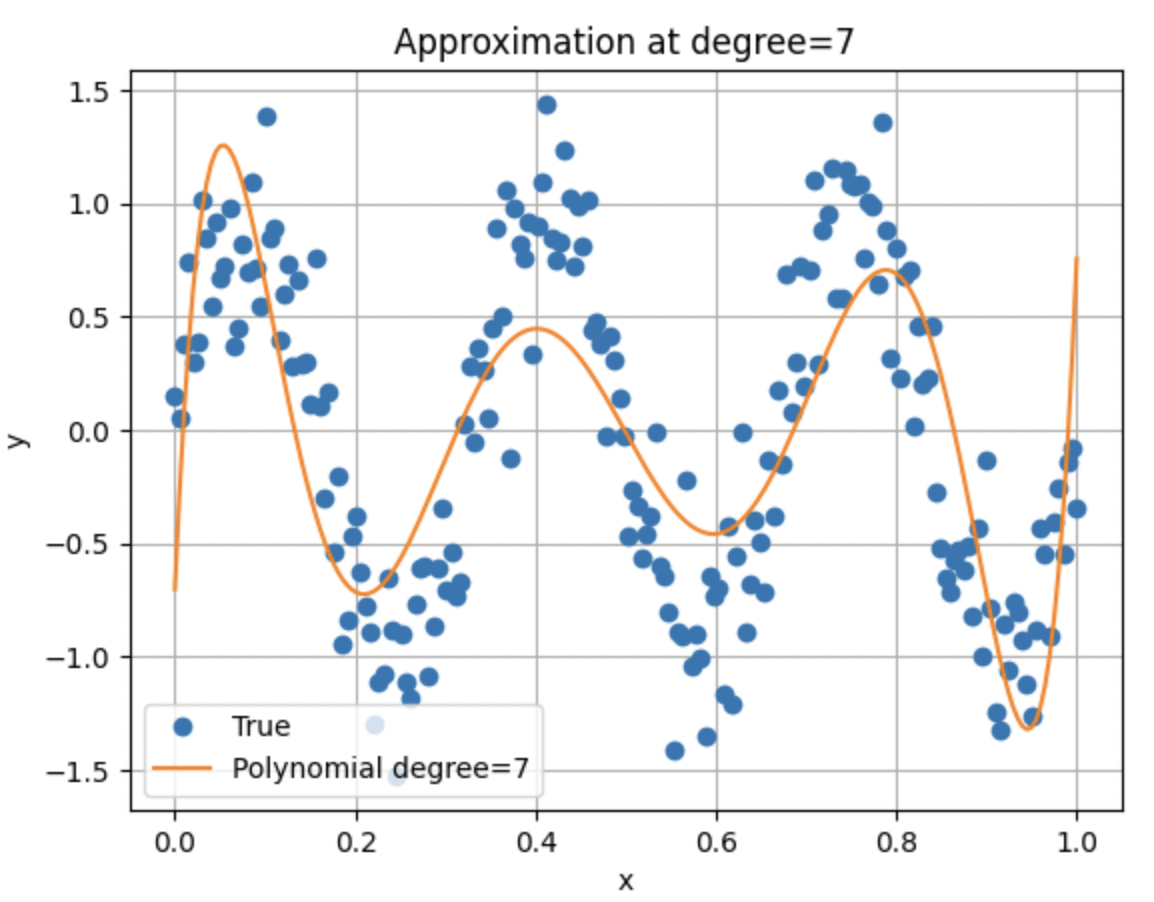}
    \caption{$d=7$, $\sigma=0.3$}
    \label{fig:poly-high-4pane-b}
  \end{subfigure}

  \vskip\baselineskip

  %––––– Bottom‐left
  \begin{subfigure}[t]{0.48\textwidth}
    \centering
    \includegraphics[width=\linewidth]{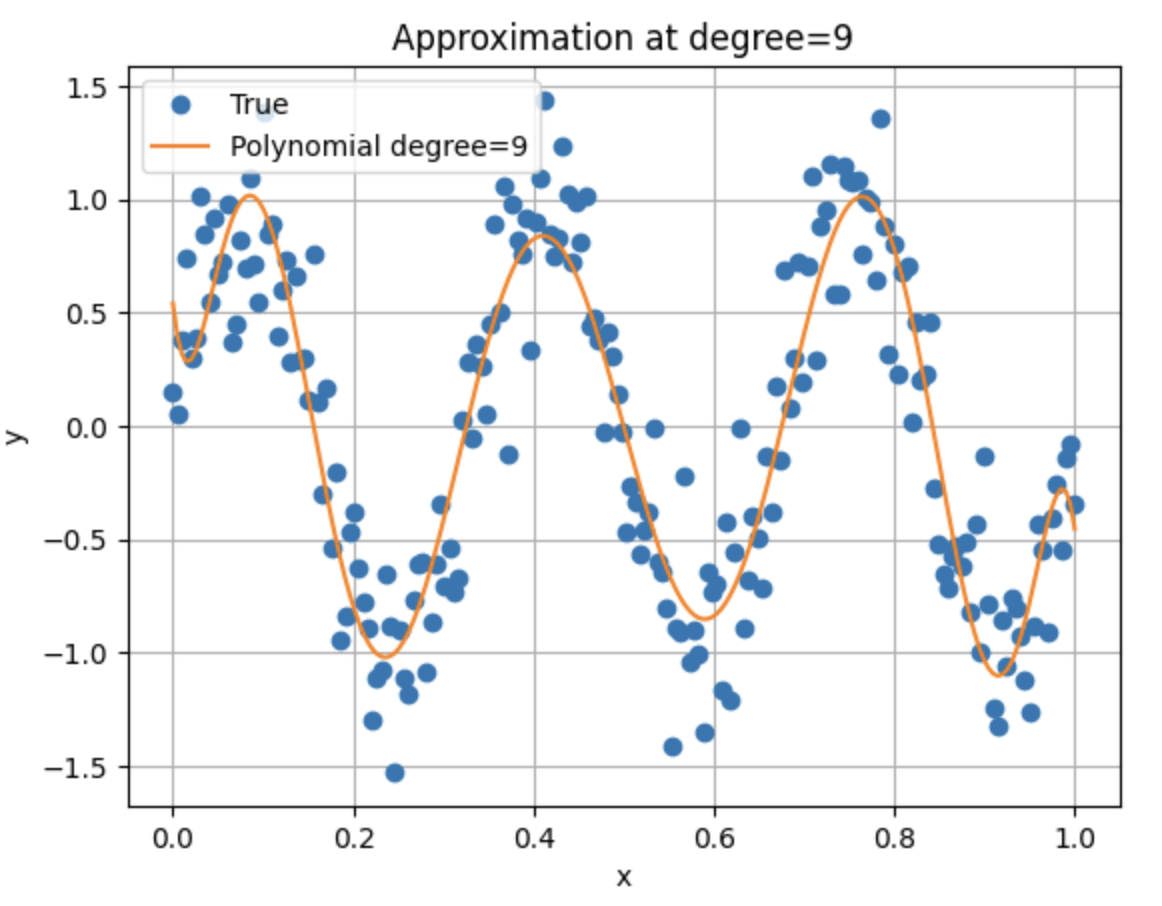}
    \caption{$d=9$, $\sigma=0.3$}
    \label{fig:poly-high-4pane-c}
  \end{subfigure}
  \hfill
  %––––– Bottom‐right
  \begin{subfigure}[t]{0.48\textwidth}
    \centering
    \includegraphics[width=\linewidth]{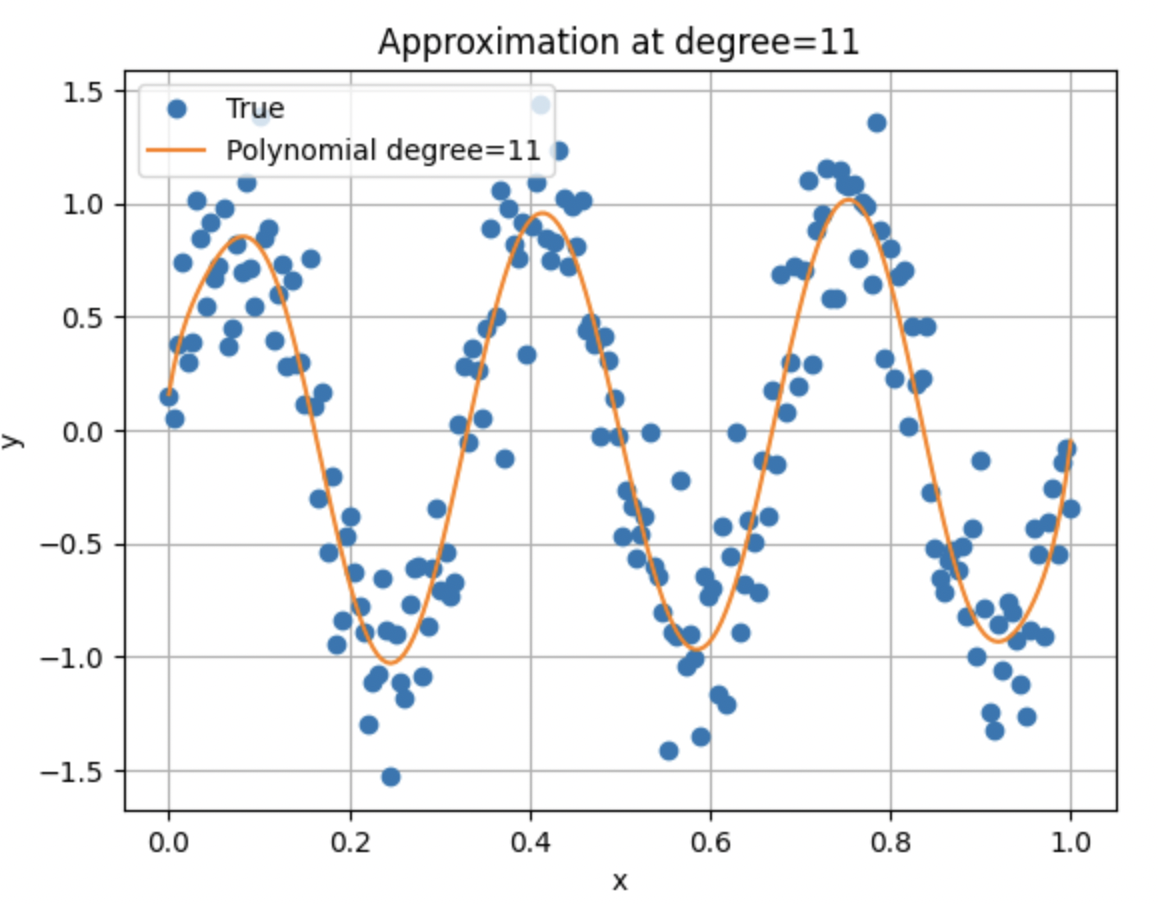}
    \caption{$d=11$, $\sigma=0.3$}
    \label{fig:poly-high-4pane-d}
  \end{subfigure}

  \caption{
  Polynomial regression for $f(x)=\sin(6\pi x) + \varepsilon$ with Gaussian noise $\varepsilon \sim N(0, \sigma^2)$ and varying polynomial degrees $d$. Note that the Loss vs. Complexity curve has ``elbows'' at $d=7$ and $d=9$. There are visible ``phase transitions'' in the shape of the polynomial vs the data  at $d = 5, 7, 9, 11$, while in between these values the regression curve shape stays relatively the same, and tends to stabilize after $d=11$.
  }
  \label{fig:poly-high-4pane}
\end{figure}

\begin{figure}[!htbp]
    \centering
    \begin{subfigure}[t]{0.45\linewidth}
        \centering
        \includegraphics[width=\linewidth]{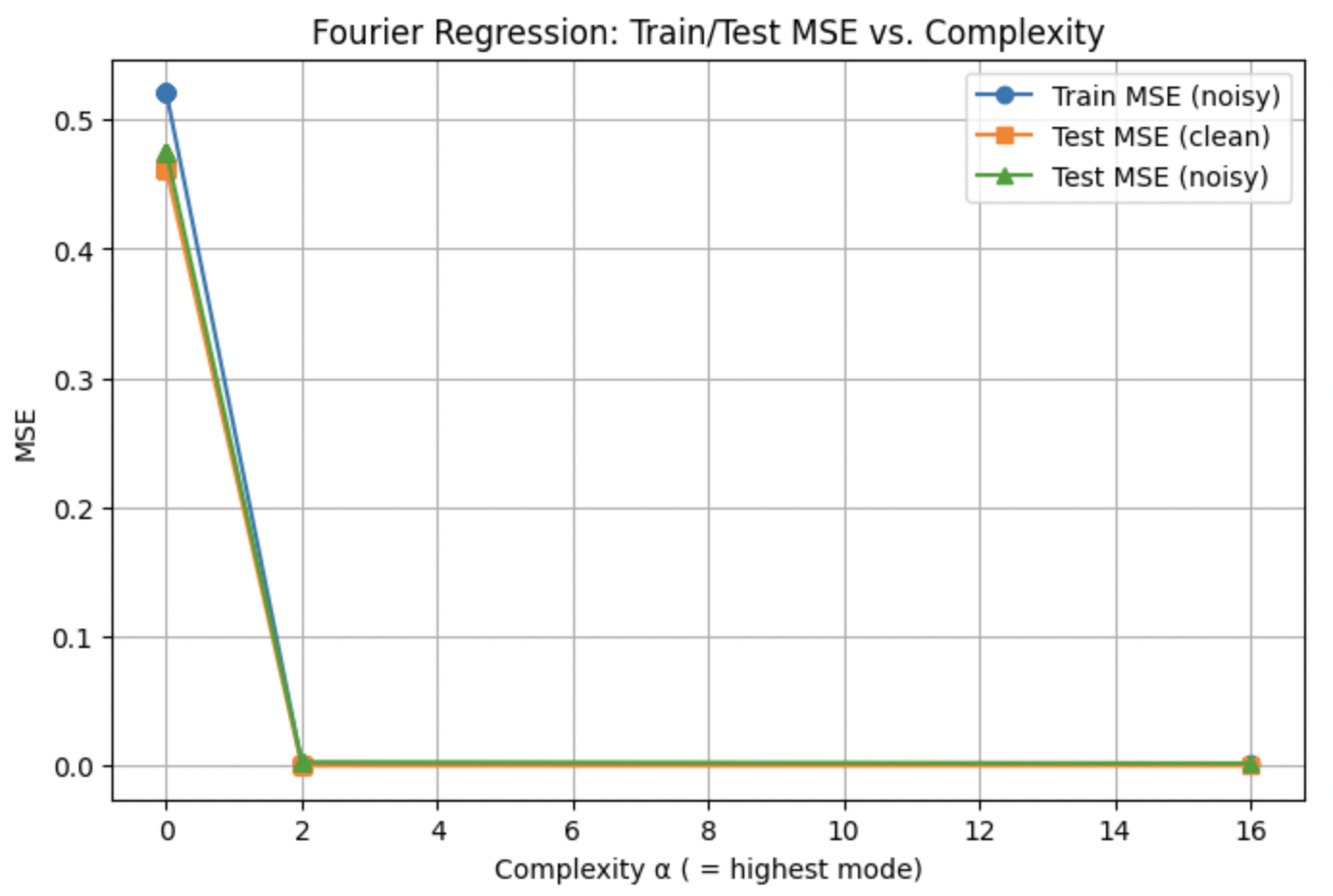}
        \caption{$\sigma = 0.05$}
        \label{fig:four-low-noise}
    \end{subfigure}
    \hfill
    \begin{subfigure}[t]{0.45\linewidth}
        \centering
        \includegraphics[width=\linewidth]{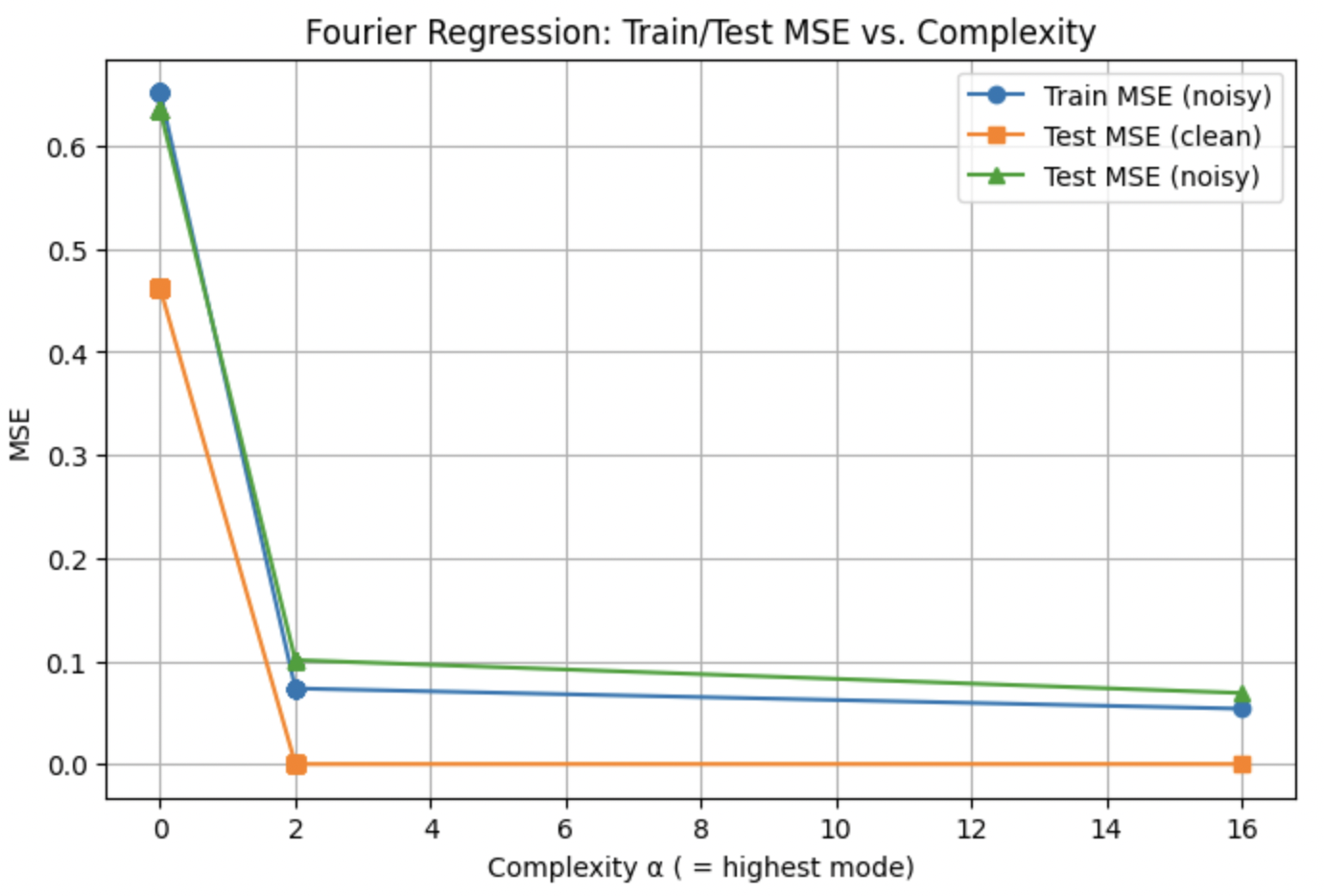}
        \caption{$\sigma=0.3$}
        \label{fig:four-high-noise}
    \end{subfigure}
    \caption{
        Loss vs.\ Complexity for polynomial regression on 
        $f(x)=\sin(4\pi x) + \varepsilon$ with Gaussian noise $\varepsilon \sim N(0, \sigma^2)$
    }
    \label{fig:four}
\end{figure}

\subsection{Tree-based models}

In the tree regressor experiment, we notice that overfitting starts right after the optimal depth. Indeed, \emph{the test loss} shows divergence after the optimal threshold. The loss function used is the usual sum of squared errors (SSE) instead of the mean squared error (MSE) only for scaling reasons: this way the overfitting threshold becomes more visible in the graphs. Here we used HyperOpt \cite{bergstra2013making} for speed as opposed to Simulated Annealing. 

The case of high noise level (Gaussian with $\mu = 0, \sigma = 0.3$) is most interesting as it shows how overfitting starts after a certain point: a salient feature is a ``dip'' instead of a simple ``elbow'' indicating a rise in the test dataset loss as the model complexity gets large enough. This can be easily seen from comparison of Figures~\ref{fig:tree-low-noise} and \ref{fig:tree-high-noise}.

\begin{figure}[!htbp]
    \centering
    \begin{subfigure}[t]{0.45\linewidth}
        \centering
        \includegraphics[width=\linewidth]{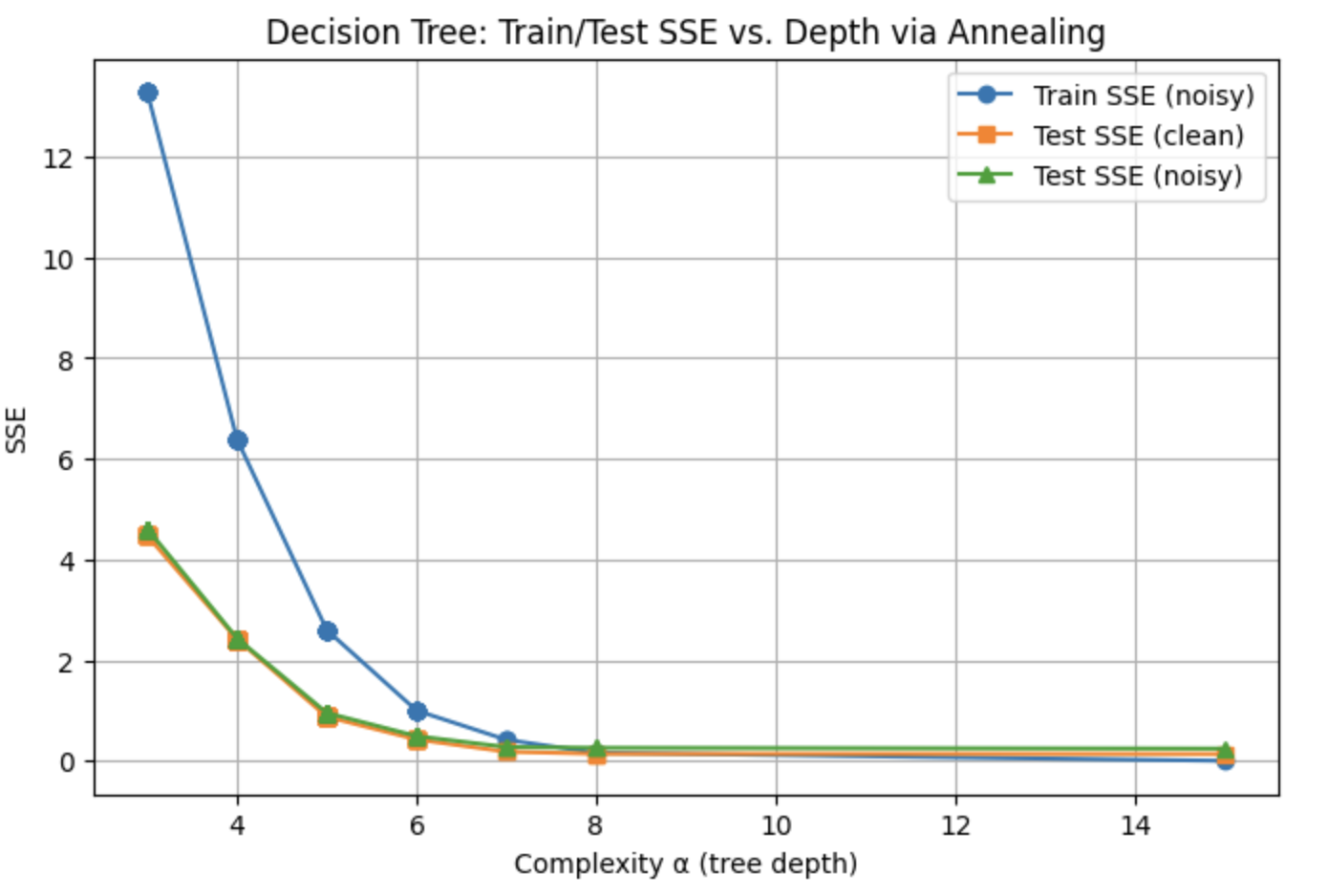}
        \caption{Simulated annealing: SSE loss computed for all datasets}
        \label{fig:anneal-low-noise}
    \end{subfigure}
    \hfill
    \begin{subfigure}[t]{0.45\linewidth}
        \centering
        \includegraphics[width=\linewidth]{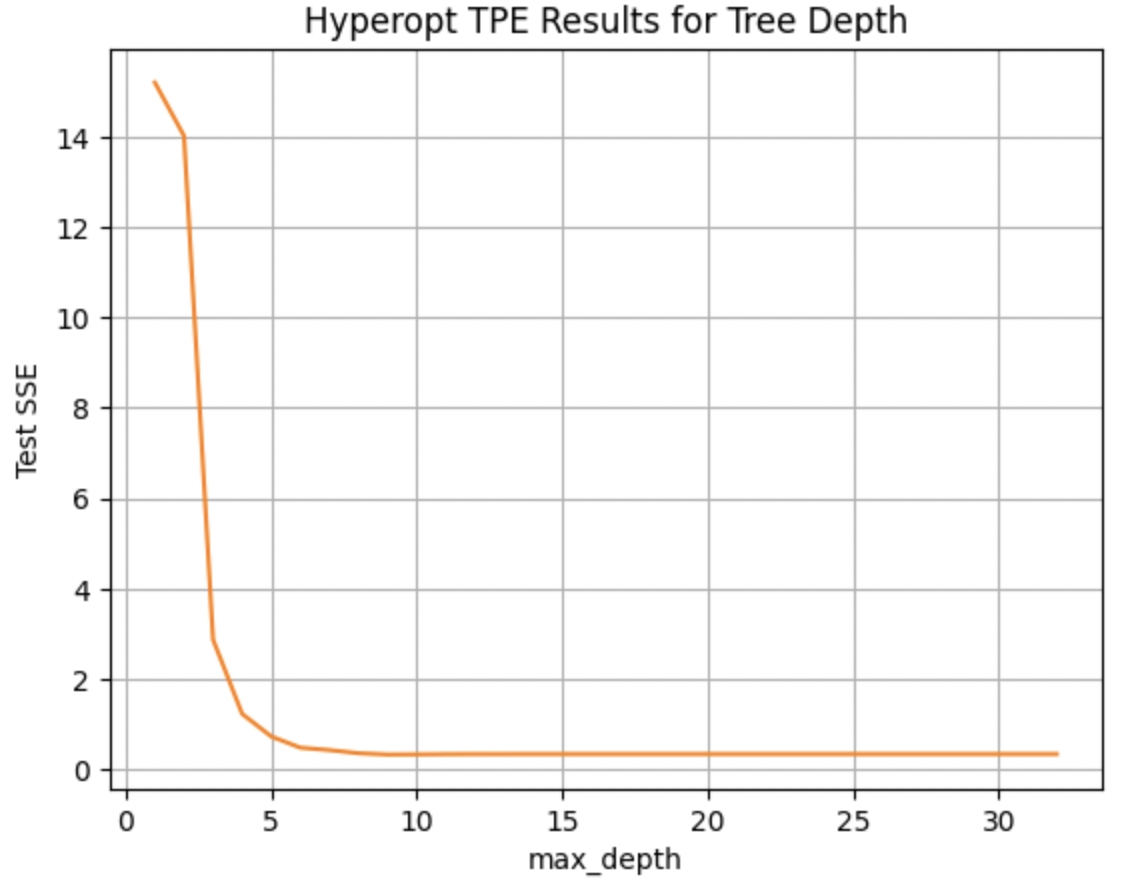}
        \caption{Tree-structured Parzen Estimator (TPE) search: noisy test dataset loss}
        \label{fig:tpe-low-noise}
    \end{subfigure}
    \caption{
        Loss vs.\ Complexity for a decision tree regressor on 
        $f(x)=\sin(4\pi x) + \varepsilon$ with Gaussian noise $\varepsilon \sim N(0, \sigma^2)$. Here $\sigma=0.05$, a low noise level case.
    }
    \label{fig:tree-low-noise}
\end{figure}

\begin{figure}[!htbp]
    \centering
    \begin{subfigure}[t]{0.45\linewidth}
        \centering
        \includegraphics[width=\linewidth]{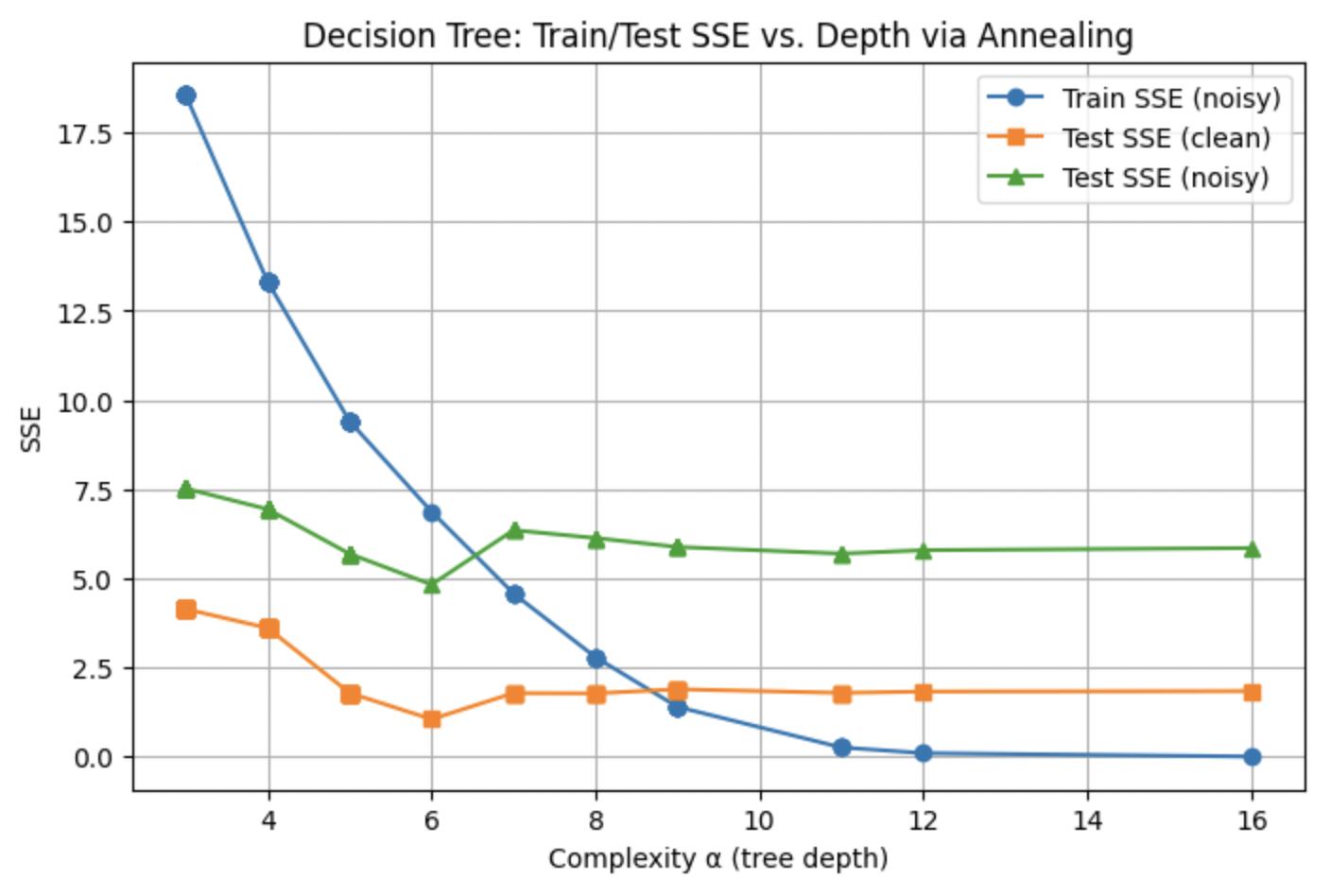}
        \caption{Simulated annealing: SSE loss computed for all datasets}
        \label{fig:anneal-high-noise}
    \end{subfigure}
    \hfill
    \begin{subfigure}[t]{0.45\linewidth}
        \centering
        \includegraphics[width=\linewidth]{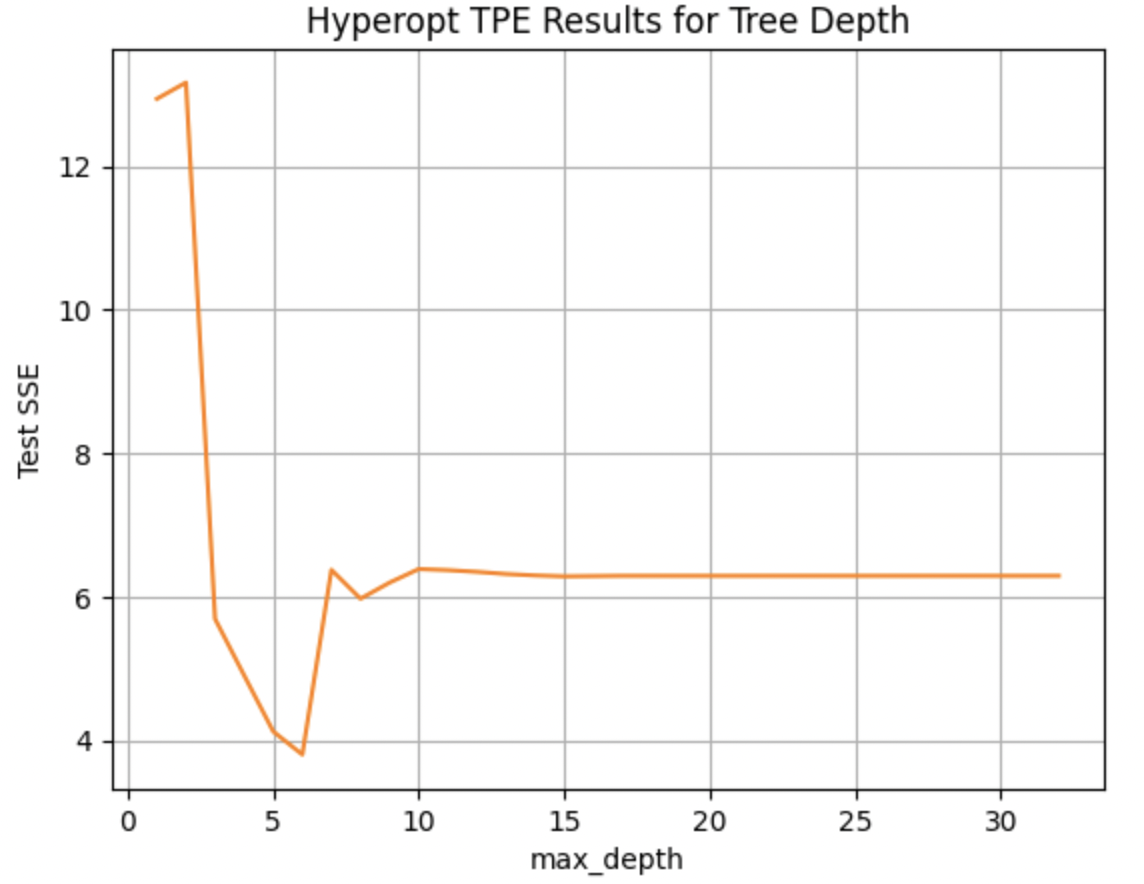}
        \caption{Tree-structured Parzen Estimator (TPE) search: noisy test dataset loss}
        \label{fig:tpe-high-noise}
    \end{subfigure}
    \caption{
        Loss vs.\ Complexity for a decision tree regressor on 
        $f(x)=\sin(4\pi x) + \varepsilon$ with Gaussian noise $\varepsilon \sim N(0, \sigma^2)$. Here $\sigma=0.3$, a high noise level case.
    }
    \label{fig:tree-high-noise}
\end{figure}

We also performed some bootstrapping experiments in order to handle the stochastic nature of tree regressors. However, the picture for the $0.95$ confidence interval consistently shows the optimal tree depth ``elbow'', as depicted in Figures~\ref{fig:tree-bootstrap-low-noise} and \ref{fig:tree-bootstrap-high-noise}. 

\begin{figure}[!htbp]
    \centering
    \begin{subfigure}[t]{0.45\linewidth}
        \centering
        \includegraphics[width=\linewidth]{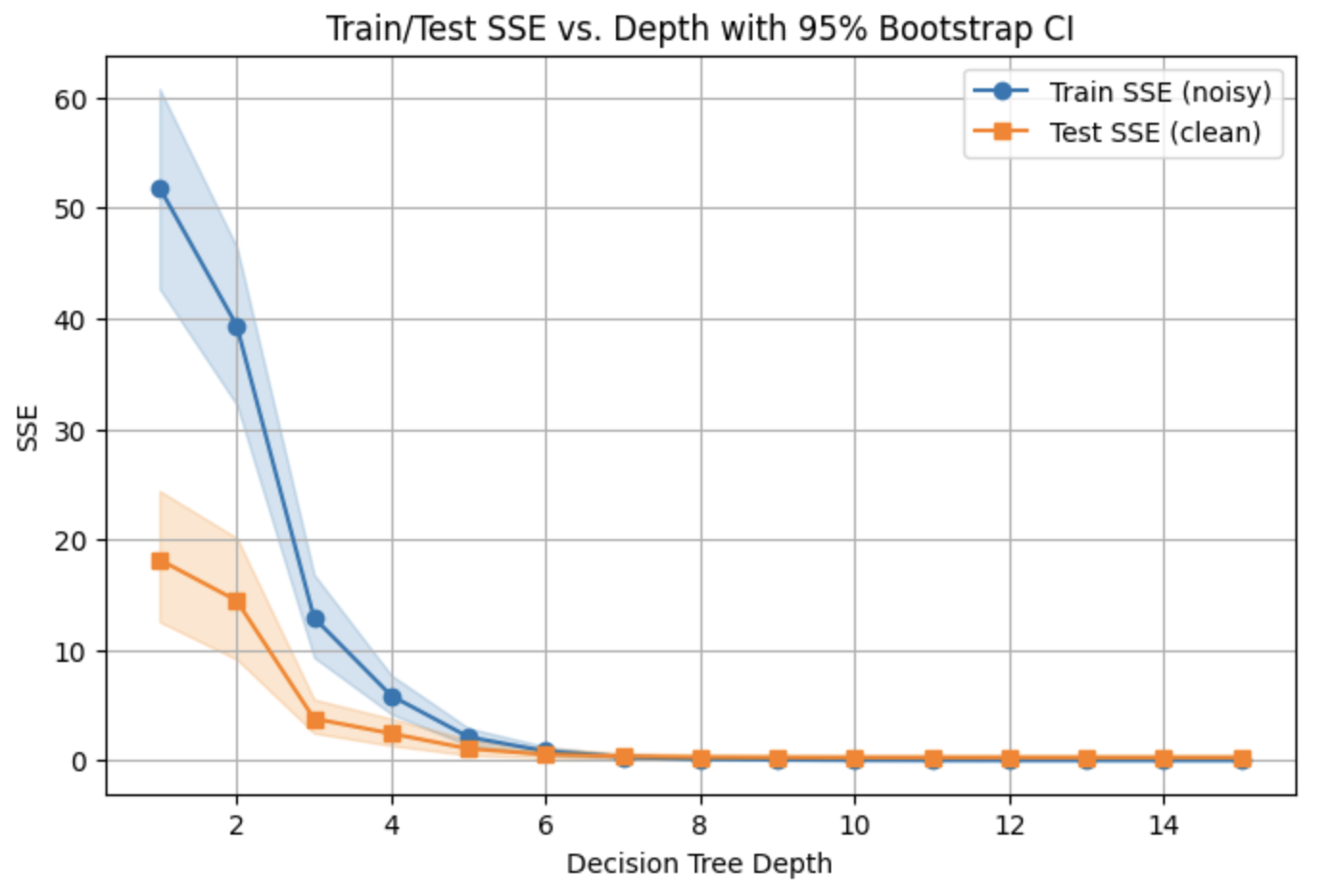}
        \caption{Loss vs. Complexity graph with $0.95$ confidence band}
        \label{fig:tpe-low-noise-graph}
    \end{subfigure}
    \hfill
    \begin{subfigure}[t]{0.45\linewidth}
        \centering
        \includegraphics[width=\linewidth]{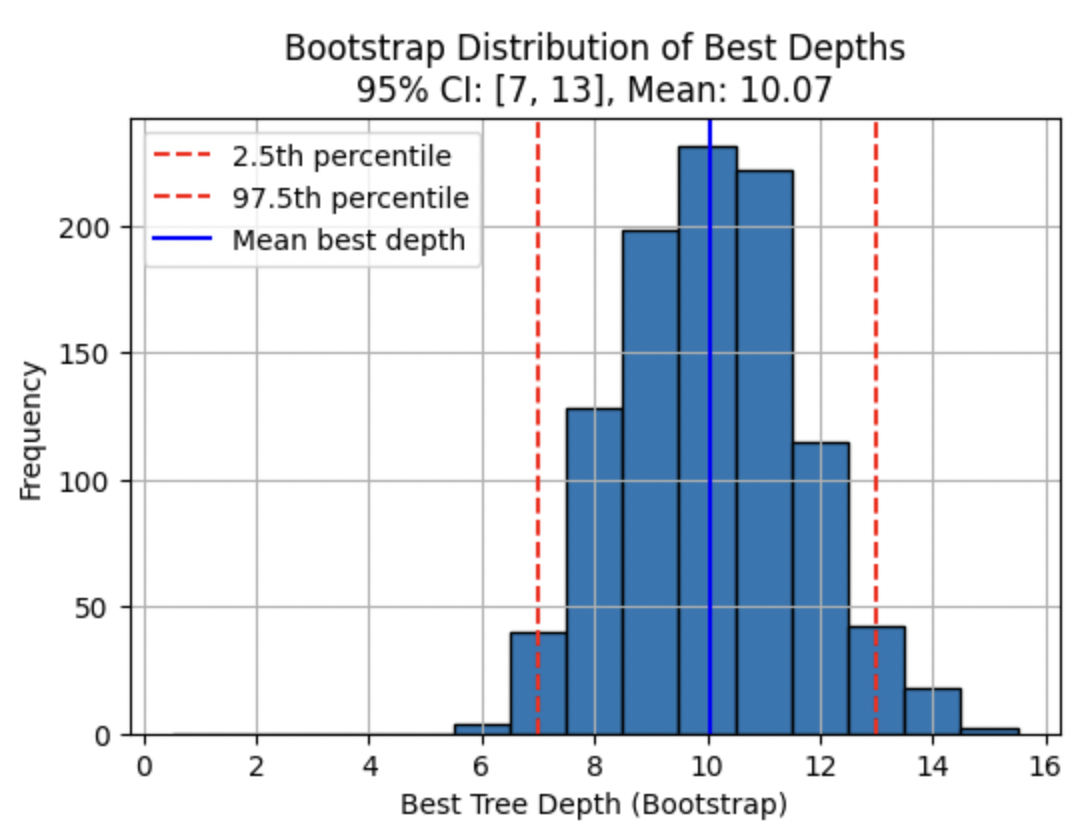}
        \caption{Histogram of optimal tree regressor depths}
        \label{fig:tpe-low-noise-ci}
    \end{subfigure}
    \caption{
        Loss vs.\ Complexity for a decision tree regressor on 
        $f(x)=\sin(4\pi x) + \varepsilon$ with Gaussian noise $\varepsilon \sim N(0, \sigma^2)$ with $\sigma=0.05$. The TPE estimator is boostrapped on $N=1000$ trials to produce $0.95$ confidence intervals. 
    }
    \label{fig:tree-bootstrap-low-noise}
\end{figure}

\begin{figure}[!htbp]
    \centering
    \begin{subfigure}[t]{0.45\linewidth}
        \centering
        \includegraphics[width=\linewidth]{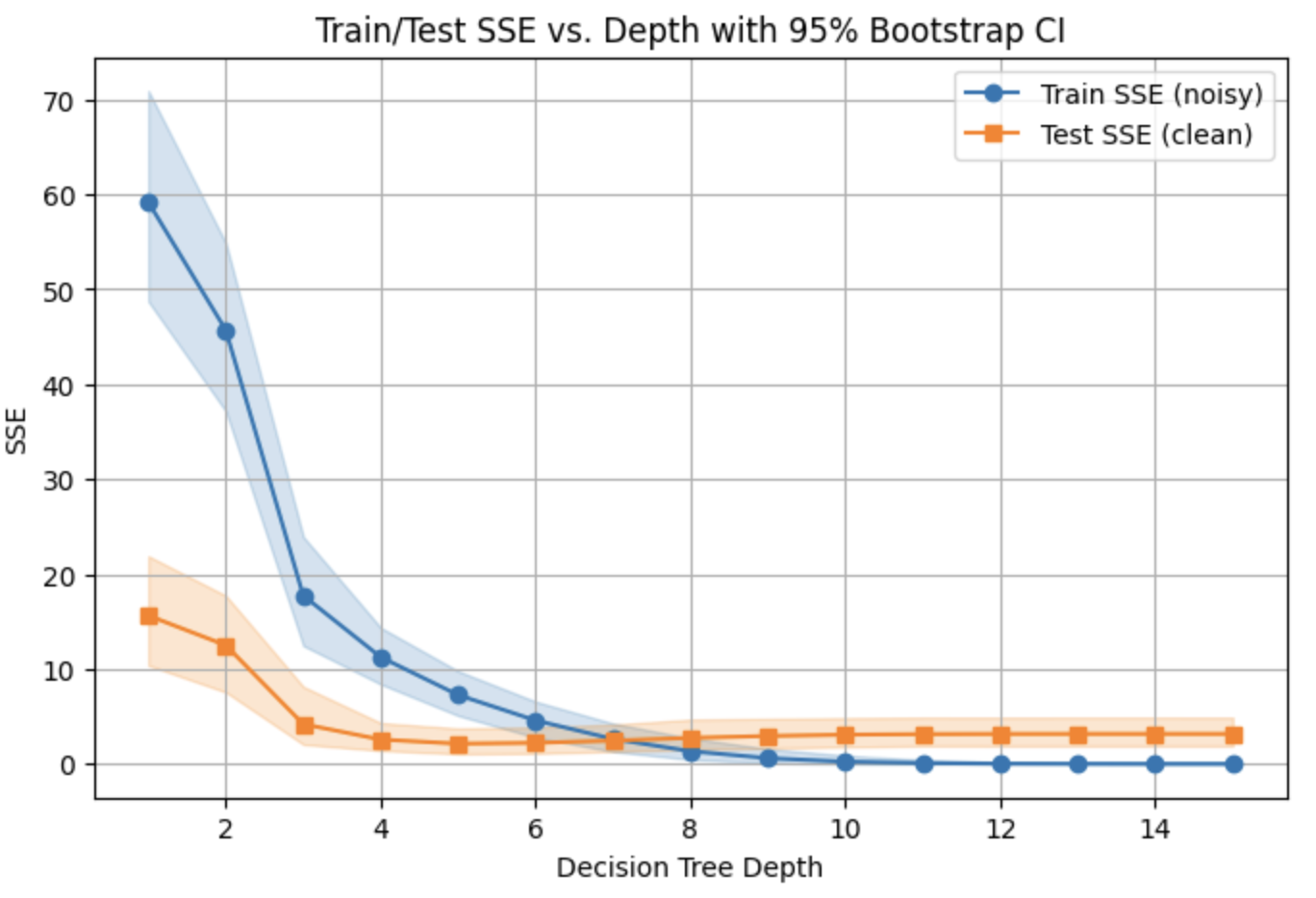}
        \caption{Loss vs. Complexity graph with $0.95$ confidence band}
        \label{fig:tpe-high-noise-graph}
    \end{subfigure}
    \hfill
    \begin{subfigure}[t]{0.45\linewidth}
        \centering
        \includegraphics[width=\linewidth]{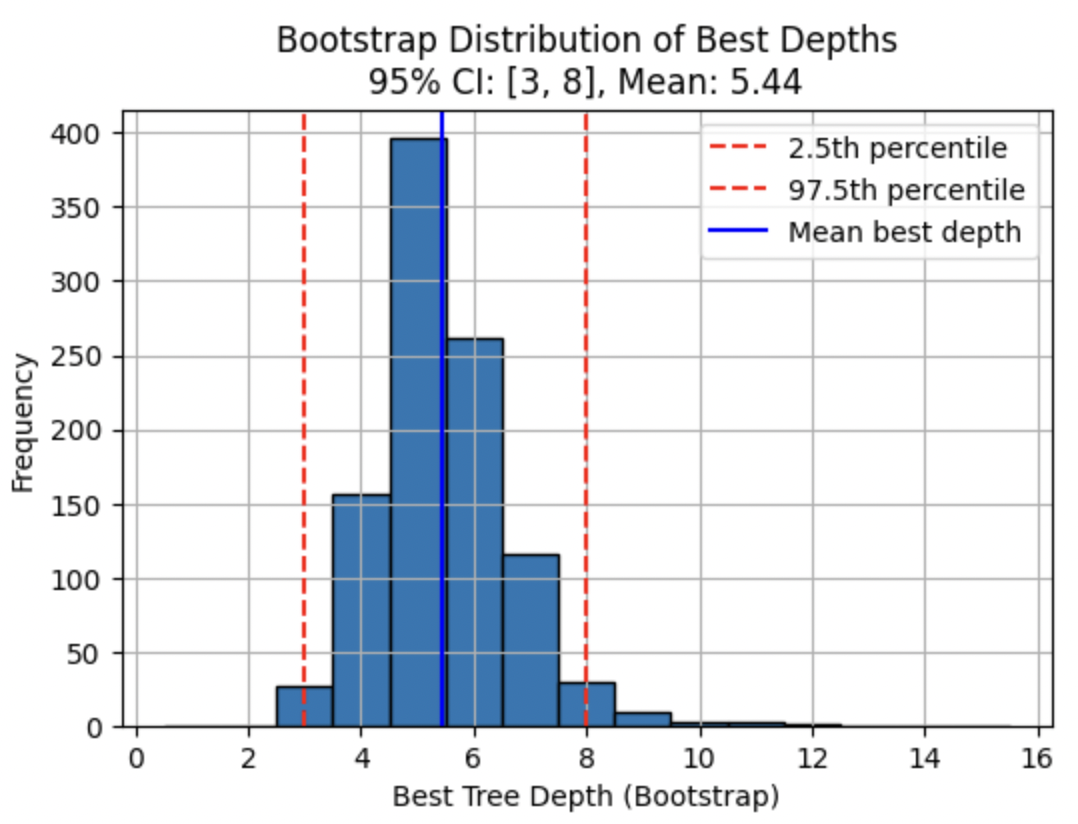}
        \caption{Histogram of optimal tree regressor depths}
        \label{fig:tpe-high-noise-ci}
    \end{subfigure}
    \caption{
        Loss vs.\ Complexity for a decision tree regressor on 
        $f(x)=\sin(4\pi x) + \varepsilon$ with Gaussian noise $\varepsilon \sim N(0, \sigma^2)$ with $\sigma=0.3$. The TPE estimator is boostrapped on $N=1000$ trials to produce $0.95$ confidence intervals. 
    }
    \label{fig:tree-bootstrap-high-noise}
\end{figure}

Also we observe the following unsurprising phenomenon: the stronger the noise, the more predictions tend to cluster as compared against the clean test dataset (see Figures~\ref{fig:preds-low-high-4pane}). The optimal depth of the tree regressor goes \emph{down} in the strong noise case, as it should be, to avoid overfitting. In the weak noise case, the optimal tree depth goes \emph{up} so that the regressor  can learn the dataset with more granularity. 

\begin{figure}[!htbp]
  \centering
  %––––– Top‐left
  \begin{subfigure}[t]{0.48\textwidth}
    \centering
    \includegraphics[width=\linewidth]{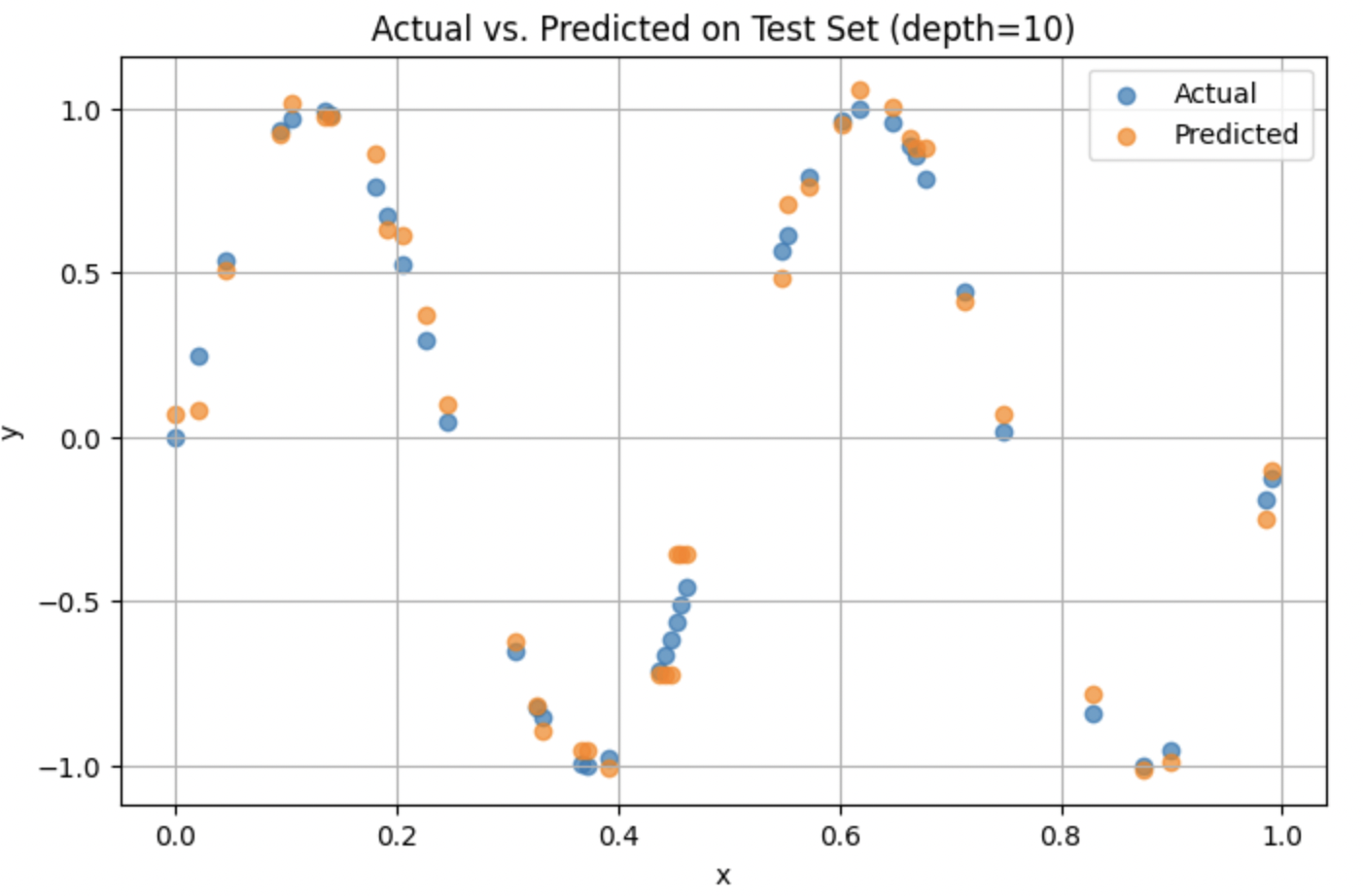}
    \caption{$\sigma=0.05$, SSE loss}
    \label{fig:tree-l2-low-noise}
  \end{subfigure}
  \hfill
  %––––– Top‐right
  \begin{subfigure}[t]{0.48\textwidth}
    \centering
    \includegraphics[width=\linewidth]{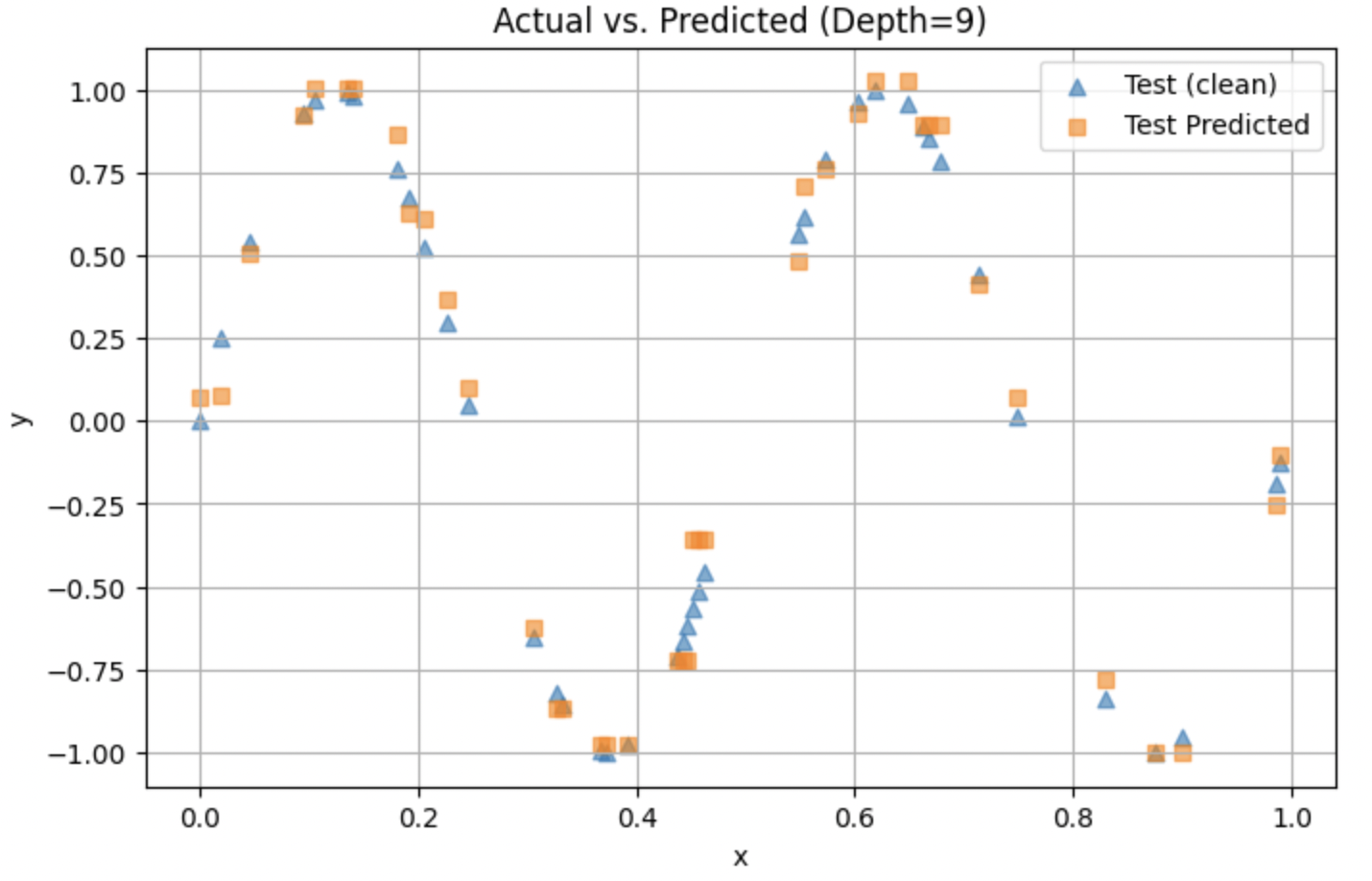}
    \caption{$\sigma=0.05$, SAE loss}
    \label{fig:tree-l1-low-noise}
  \end{subfigure}

  \vskip\baselineskip

  %––––– Bottom‐left
  \begin{subfigure}[t]{0.48\textwidth}
    \centering
    \includegraphics[width=\linewidth]{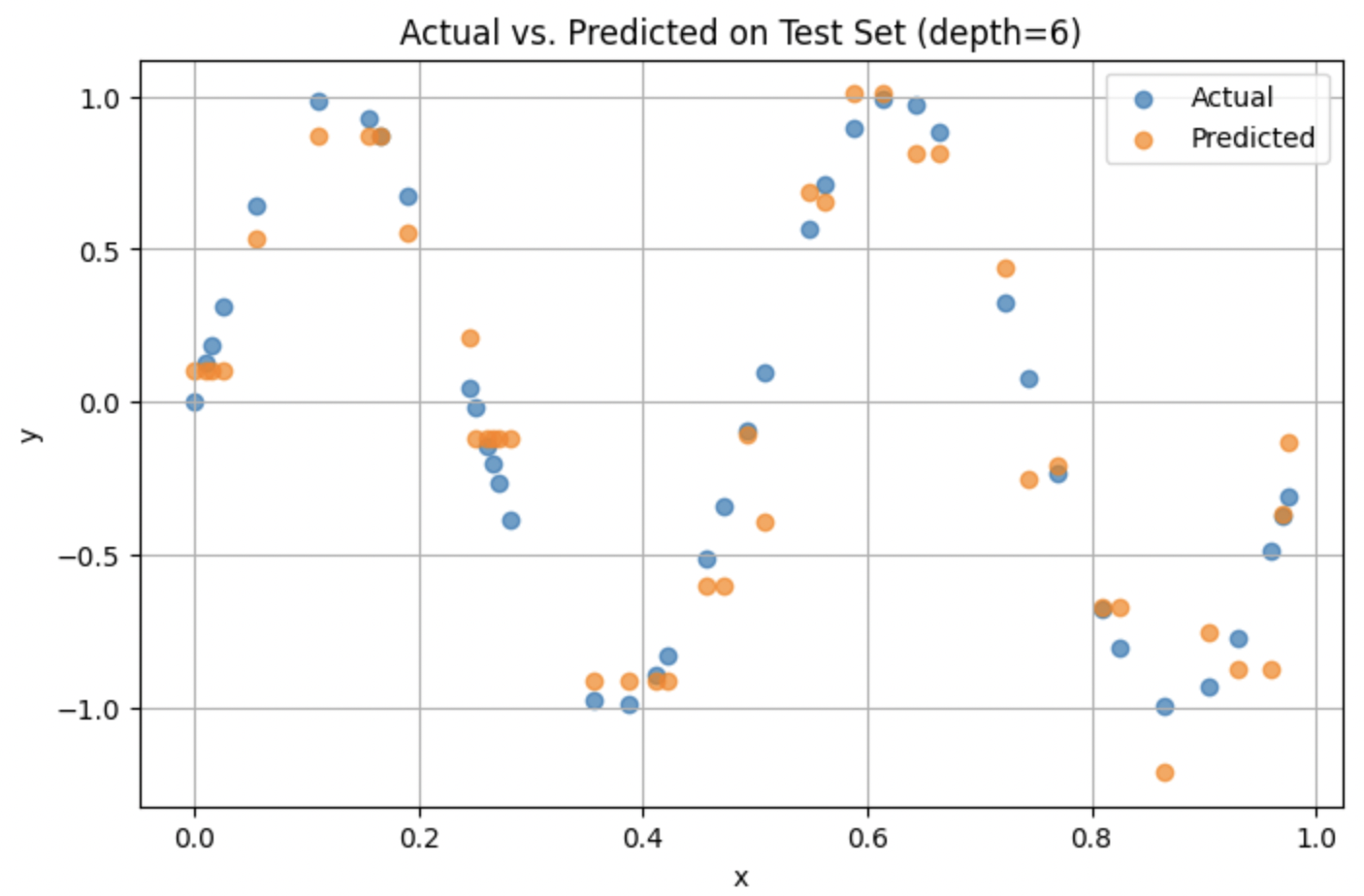}
    \caption{$\sigma=0.3$, SSE loss}
    \label{fig:tree-l2-high-noise}
  \end{subfigure}
  \hfill
  %––––– Bottom‐right
  \begin{subfigure}[t]{0.48\textwidth}
    \centering
    \includegraphics[width=\linewidth]{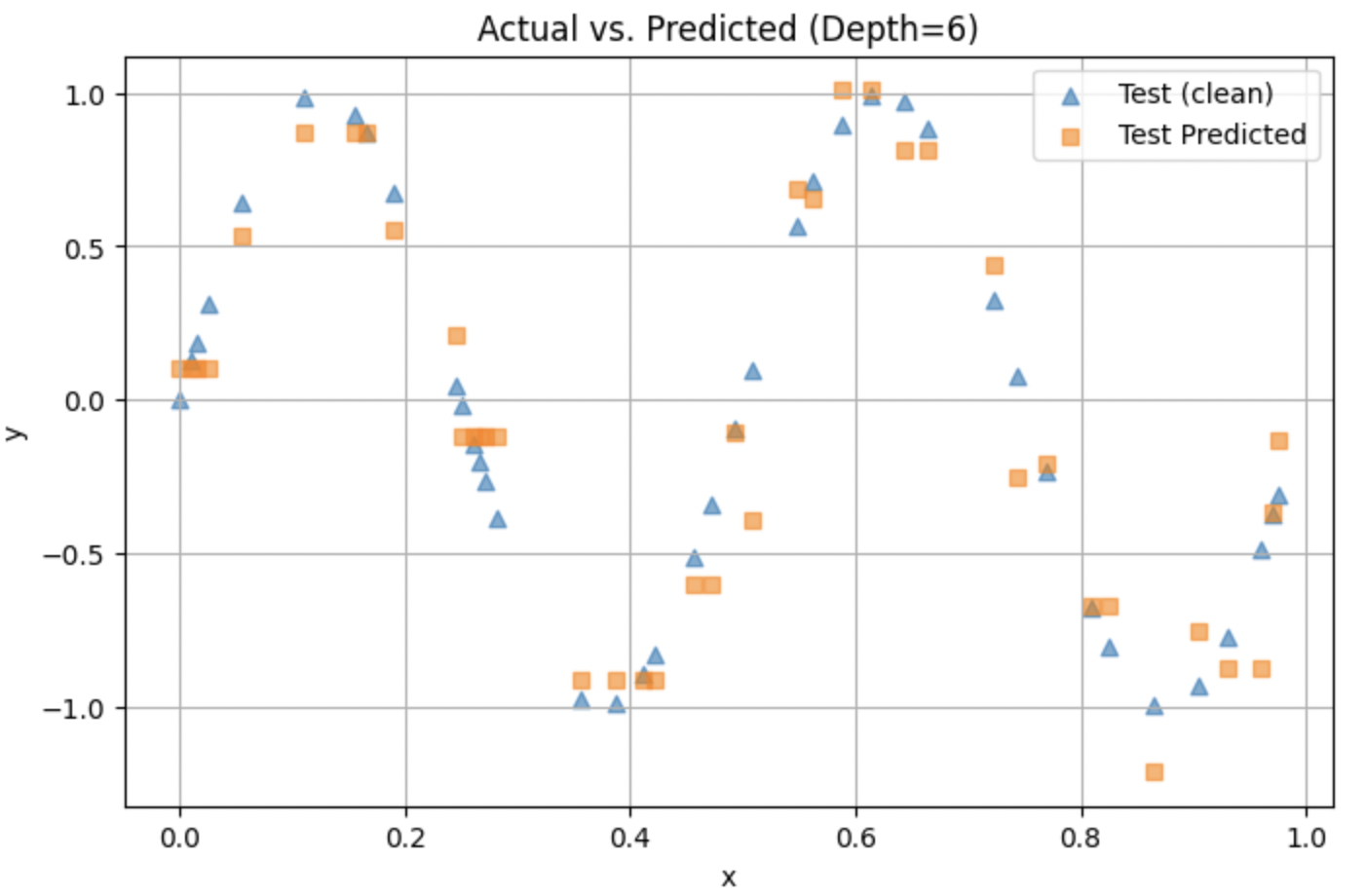}
    \caption{$\sigma=0.3$, SAE loss}
    \label{fig:tree-l1-high-noise}
  \end{subfigure}

  \caption{Predictions vs. clean test data for a decision tree regressor on $f(x)=\sin(4\pi x) + \varepsilon$ with Gaussian noise $\varepsilon \sim N(0, \sigma^2)$. The train / test loss is either SSE ($L_2$) or SAE ($L_1$). 
  }
  \label{fig:preds-low-high-4pane}
\end{figure}

Here we use the standard \textrm{DecisionTreeRegressor} class from  \textrm{sklearn}. The computation can be fully reproduced in the Google Colab environment \cite{github-structure}. 

\subsection{Deep neural networks}

A more complicated example of a neural network based on a Directed Acyclic Graph (DAG) is provided in \cite{github-structure}. This DAG represents a network with fully connected layers followed by ReLU nonlinearities. The loss function used is MSE, while the network complexity is composite: it accounts for both the topology and training hyperparameters such as the learning rate and number of epochs. Namely, for a given DAG $D$, learning rate $\lambda$ and number of epochs $N$, we have
\[
\text{Comp}(D) = E(D) \cdot (1 + \text{AvgClustering}(D))\cdot \text{ASP}(D),
\]
where $E(D)$ is the number of edges of $D$, $\text{AvgClustering}$ is the average clustering coefficient, and $\text{ASP}$ is the average shortest path between any pair of vertices connected by directed edges of $D$. 

Let $M = (D, \lambda, N)$ be the model based on $D$ with learning rate $\lambda$ and the number of training epochs $N$. Then the model \emph{composite complexity} equals
\[
\text{Comp}(M) = \text{Comp}(D) + \frac{1}{\lambda} + N. 
\]

In Figures~\ref{fig:deep-network-low-noise} and \ref{fig:deep-network-high-noise} we picture the Pareto frontier of \texttt{HyperOpt} search where the best fit (lowest MSE) model is marked, as well as the most salient ``elbow'' point (maximum distance from the line joining Pareto frontier's endpoints). 

A few fits other than the best fit are shown: a low complexity model, a high complexity model, and the most salient ``elbow'' point model. We can see that both low and high complexity are lacking goodness-of-fit (in the high complexity case, possibly because the learning rate is too small). The elbow point provide a fit that already resembles the best one, as the phase transition happens after which the model gains in complexity to improve the fit further while keeping the same qualitative behavior. 

Here we remark that the goodness-of-fit displayed by the best models in Figure~\ref{fig:deep-network-low-noise} (weak noise, $\sigma=0.05$) and Figure~\ref{fig:deep-network-high-noise} (strong noise, $\sigma=0.3$) are comparable while in the presence of strong noise the model complexity required raises by a factor of $\approx 2$ in our numerical experiments. The reader is welcome to reproduce them by running the Jupyter notebook available from \cite{github-structure} on Google Colab.  

\begin{figure}[!htbp]
    \centering
    \begin{subfigure}[t]{0.45\linewidth}
        \centering
        \includegraphics[width=\linewidth]{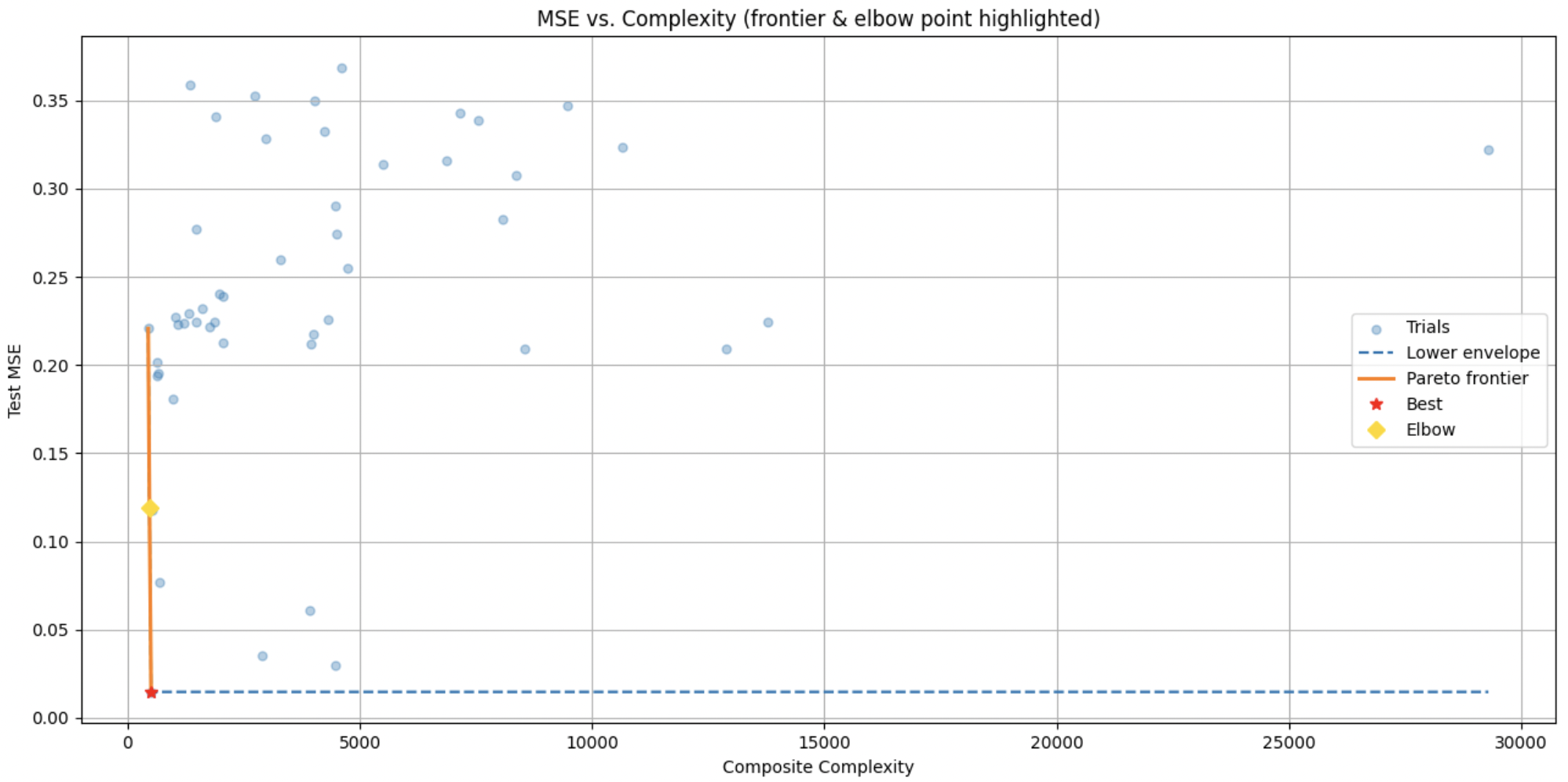}
        \caption{Pareto frontier and lower envelope}
        \label{fig:tpe-deep-low-noise-graph}
    \end{subfigure}
    \hfill
    \begin{subfigure}[t]{0.45\linewidth}
        \centering
        \includegraphics[width=\linewidth]{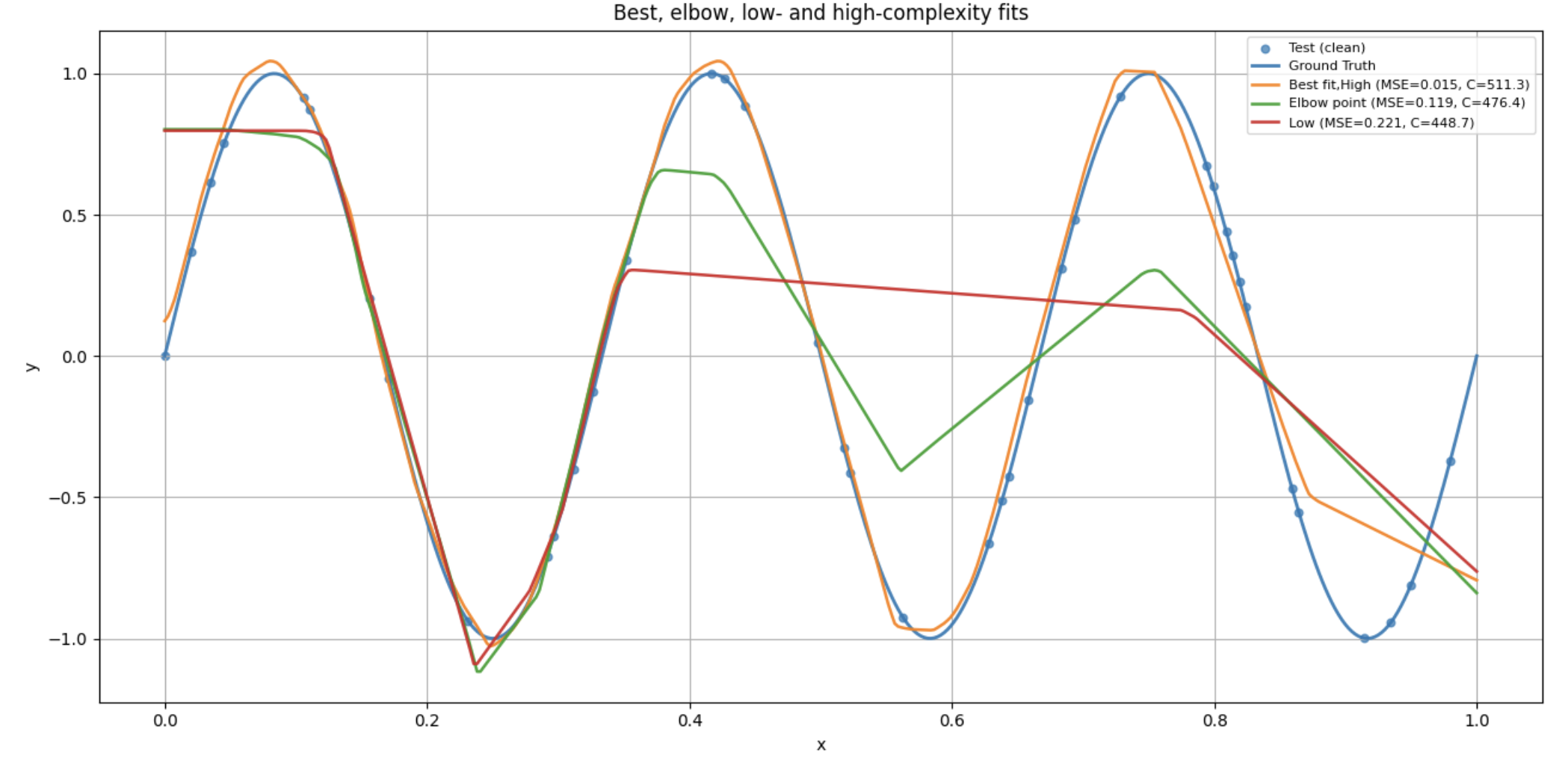}
        \caption{Model fit at various complexities}
        \label{fig:tpe-deep-low-noise-fit}
    \end{subfigure}
    \caption{
        Loss vs.\ Complexity for a deep network regressor on 
        $f(x)=\sin(6\pi x) + \varepsilon$ with Gaussian noise $\varepsilon \sim N(0, \sigma^2)$ with $\sigma=0.05$. Since training a ``deep'' network results in higher test SSE variability, we also show the Pareto boundary and lower convex envelope to illustrate the elbow point. 
    }
    \label{fig:deep-network-low-noise}
\end{figure}

\begin{figure}[!htbp]
    \centering
    \begin{subfigure}[t]{0.45\linewidth}
        \centering
        \includegraphics[width=\linewidth]{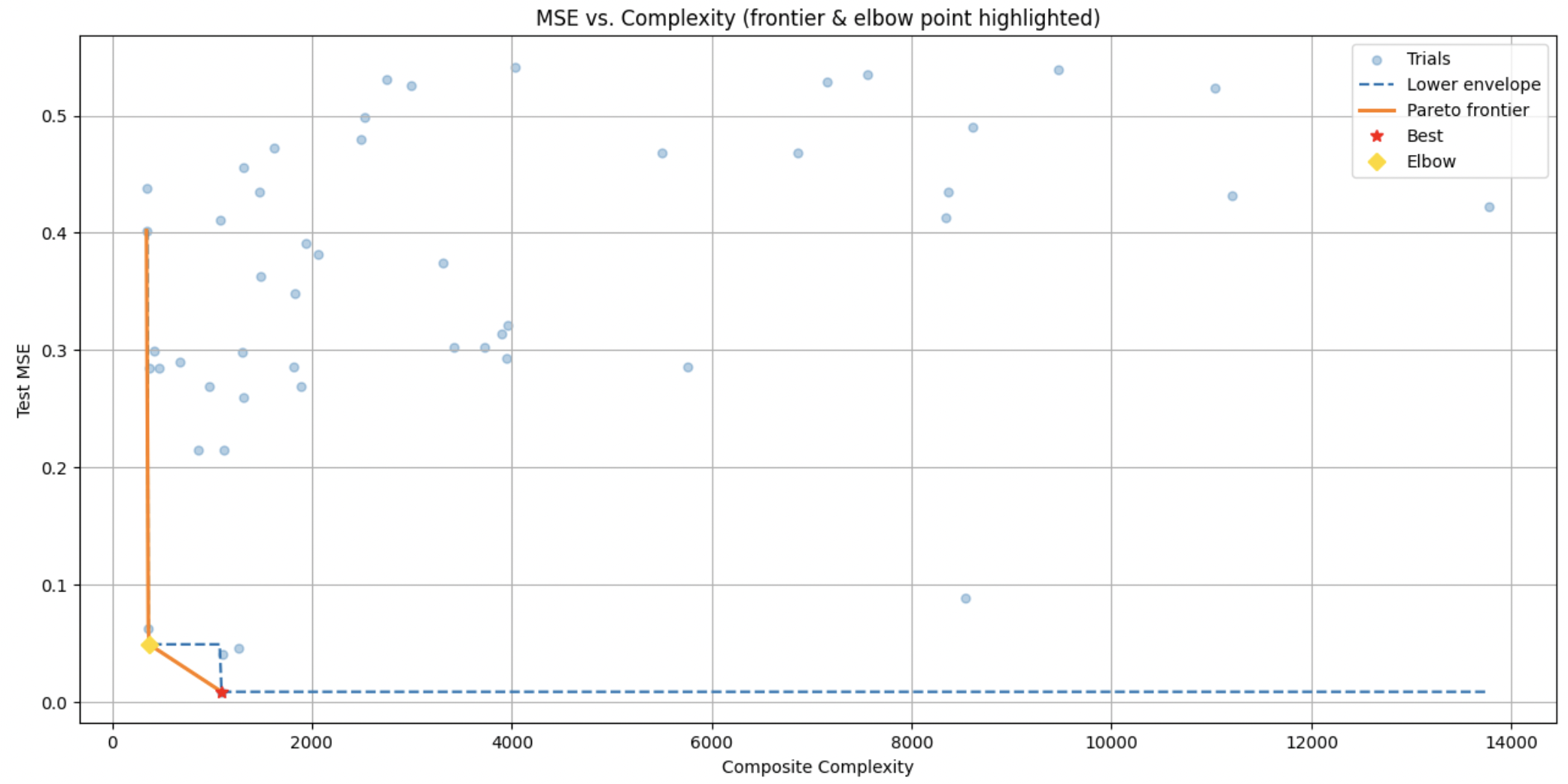}
        \caption{Loss vs. Complexity Pareto frontier}
        \label{fig:tpe-deep-high-noise-graph}
    \end{subfigure}
    \hfill
    \begin{subfigure}[t]{0.45\linewidth}
        \centering
        \includegraphics[width=\linewidth]{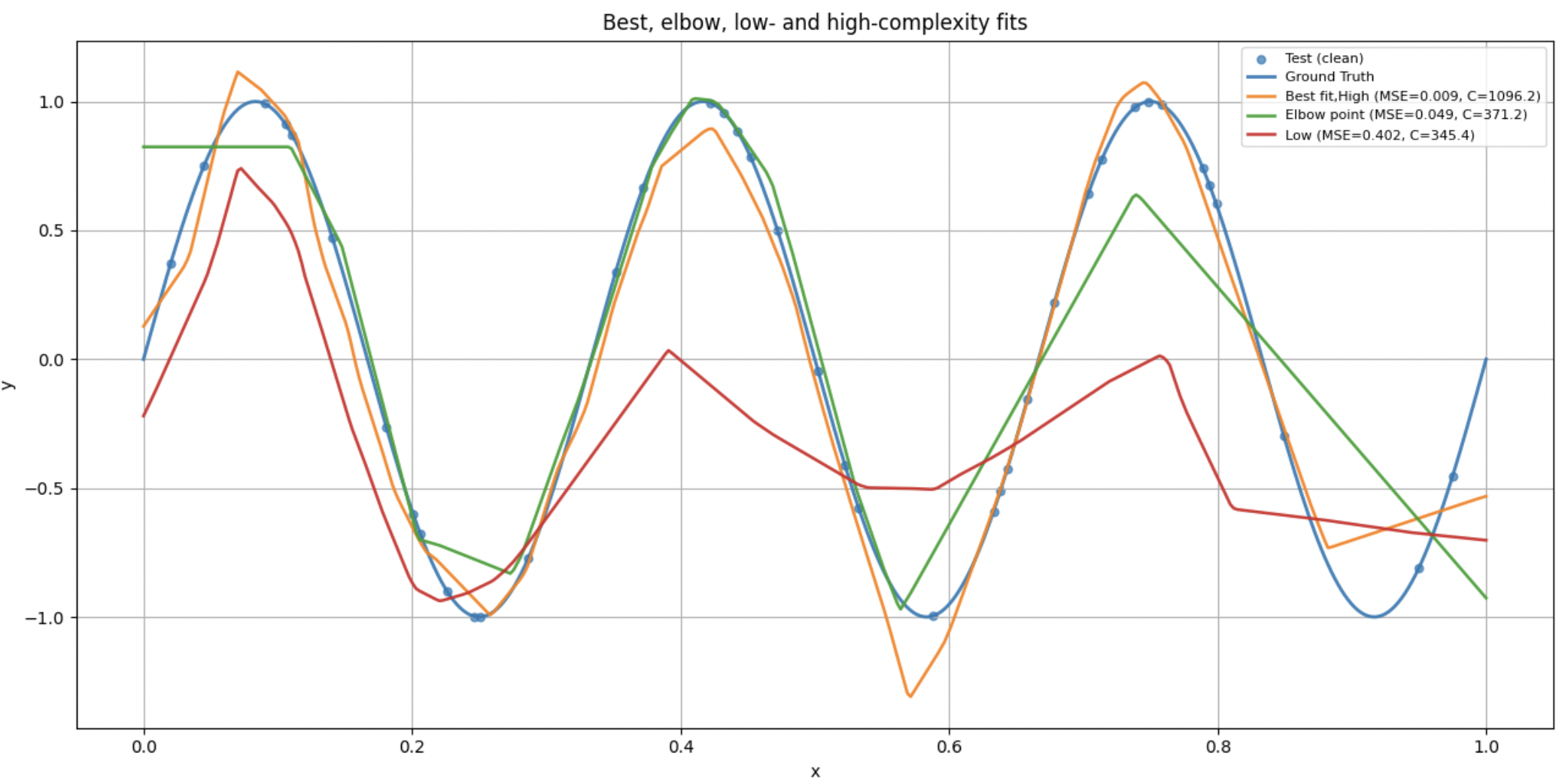}
        \caption{Model fit at various complexities}
        \label{fig:tpe-deep-high-noise-fit}
    \end{subfigure}
    \caption{
        Loss vs.\ Complexity for a deep network regressor on 
        $f(x)=\sin(6\pi x) + \varepsilon$ with Gaussian noise $\varepsilon \sim N(0, \sigma^2)$ with $\sigma=0.3$. The Pareto boundary and lower convex envelope of training SSE points are also shown to illustrate the elbow point. 
    }
    \label{fig:deep-network-high-noise}
\end{figure}

\subsection{Using different Complexity Functionals}\label{subsec:comp-practice}

As discussed in Section~\ref{subsec:comp-design}, $\mathrm{Comp}(S)$ is not a ``universal'' quantity but an operational choice: it should measure whichever resource or inductive bias we actually care about (e.g.\ interpretability, latency, or description length). Consequently the ``elbow'' seen in the dual free energy
\[
F(\lambda)=\min_S\bigl[\mathrm{Loss}(S)+\lambda\,\mathrm{Comp}(S)\bigr]
\]
and the associated optimizer path $\alpha(\lambda)$ are \emph{proxy-dependent} objects.

Our tree-regression experiment \cite{github-structure} makes this dependence concrete by comparing three proxies on the same hypothesis class: (i) classical tree depth, (ii) a compression proxy given by the \texttt{zlib}-compressed byte length of a serialized tree, and (iii) an ``effective partition'' proxy given by the leaf-mass entropy
\[
\mathrm{Comp}_{\mathrm{ent}}(S)\;=\;-\sum_{\ell\in\mathrm{leaves}} p_\ell\log p_\ell,
\qquad
p_\ell=\frac{n_\ell}{N_{\mathrm{train}}}.
\]
Because these proxies live on very different numerical scales (depth $\sim 10$, entropy $\sim 5$, \texttt{zlib} length $\sim 10^3$ bytes), we compare them using a normalized multiplier $\widetilde\lambda$ (one per proxy) so that the complexity term is commensurate with the loss term. With this normalization, cross-proxy comparisons remain meaningful and can be performed visually, see Figures \ref{fig:sse-structure-depth}, \ref{fig:sse-structure-zlib} and \ref{fig:sse-structure-entropy}. 

\begin{figure}[!htbp]
    \centering
    \includegraphics[width=\linewidth]{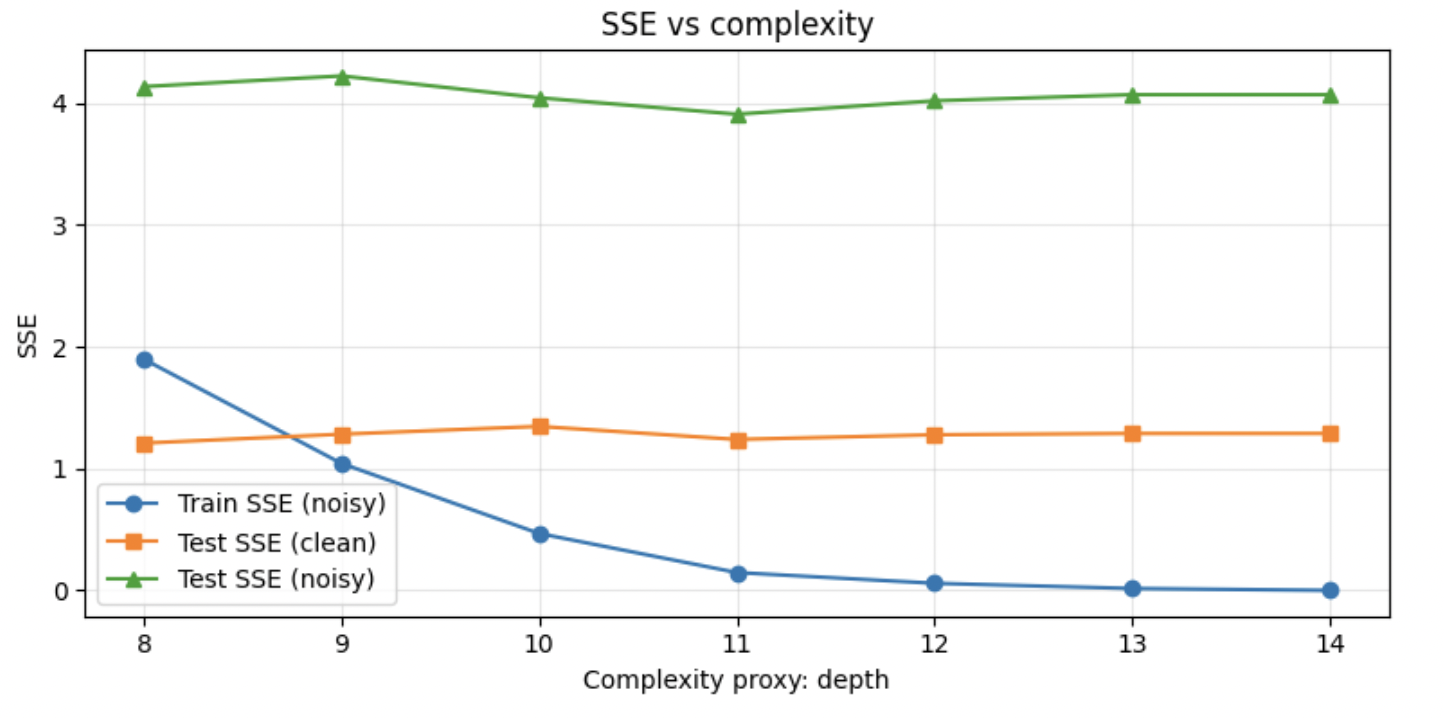}
    \caption{Loss--Complexity landscape: SSE vs $\mathrm{Comp}(S) = \text{tree-depth of } S$.}
    \label{fig:sse-structure-depth}
\end{figure}

\begin{figure}[!htbp]
    \centering
    \includegraphics[width=\linewidth]{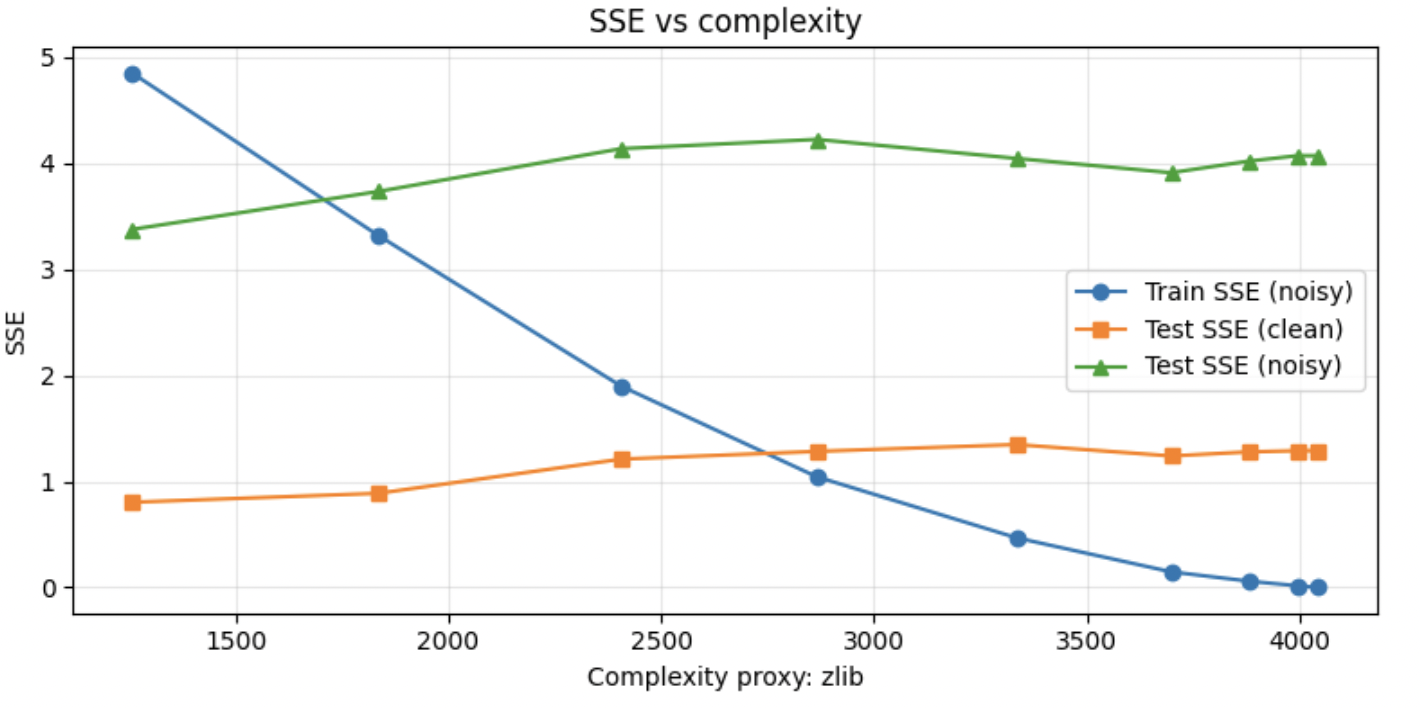}
    \caption{Loss--Complexity landscape: SSE vs $\mathrm{Comp}(S) = \text{byte-length of zlib compressed } S$.}
    \label{fig:sse-structure-zlib}
\end{figure}

\begin{figure}[!htbp]
    \centering
    \includegraphics[width=\linewidth]{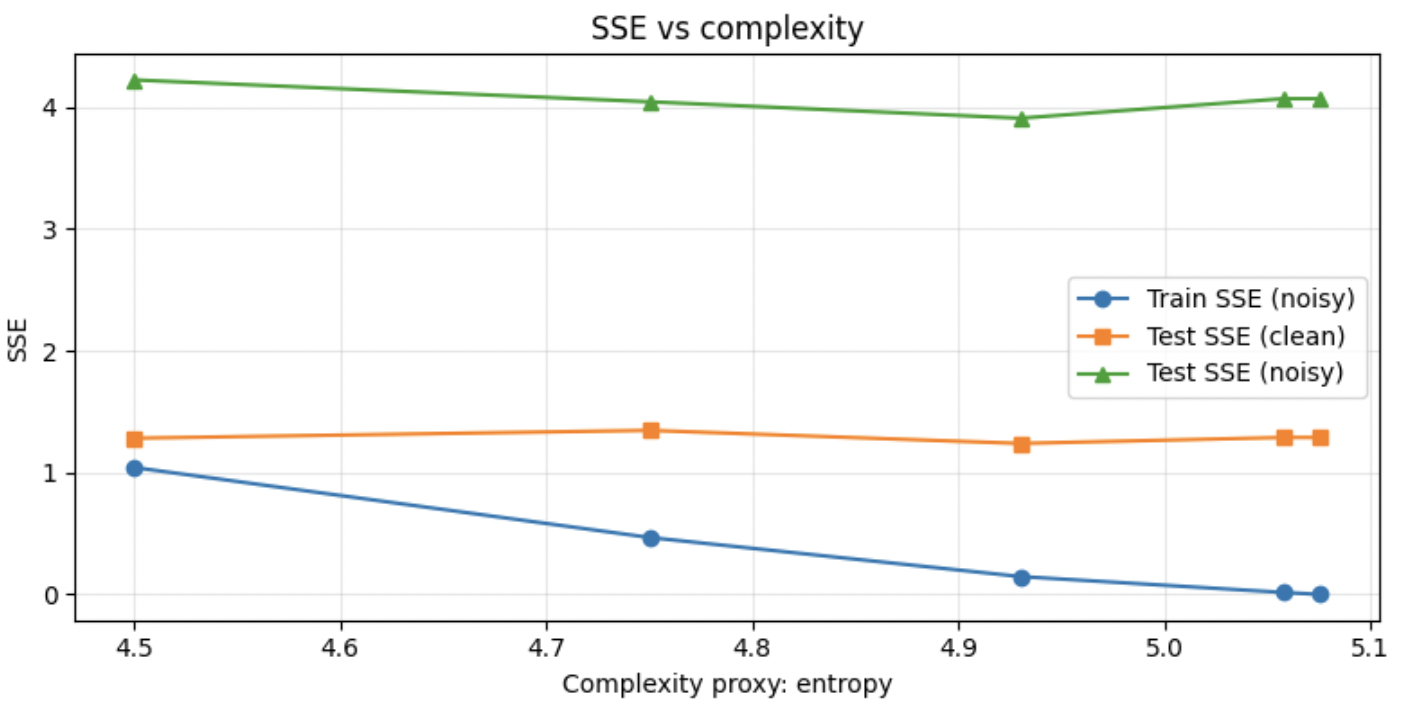}
    \caption{Loss--Complexity landscape: SSE vs $\mathrm{Comp}(S) = \text{leaf-entropy of } S$.}
    \label{fig:sse-structure-entropy}
\end{figure}

With this normalization, the $\widetilde\lambda=0$ sanity check behaves exactly as the theory predicts: since the action reduces to $\mathrm{Loss}(S)$, all proxies select the same (pure training-loss) minimizer, and indeed the depth, \texttt{zlib}, and entropy runs return the identical tree at $\tilde\lambda=0$. 

For $\widetilde\lambda>0$, however, the optimizer paths $\alpha(\widetilde\lambda)$ diverge: the three proxies produce visibly different loss--complexity landscapes and, more importantly, select different best-generalizing trees. In the reported run, the model minimizing clean test SSE under the depth proxy occurs at depth $8$ (test SSE$_{\mathrm{clean}}\approx 1.21$, see Figure~\ref{fig:sse-structure-depth}), the \texttt{zlib} proxy favors a much shallower depth-$6$ tree (test SSE$_{\mathrm{clean}}\approx 0.80$, see Figure~\ref{fig:sse-structure-zlib}), while the entropy proxy selects a deeper depth-$11$ tree (test SSE$_{\mathrm{clean}}\approx 1.24$, see Figure~\ref{fig:sse-structure-entropy}); the corresponding tree signatures are pairwise distinct \cite{github-structure}, confirming that these are genuinely different minimizers rather than reparametrizations.

These discrepancies reflect real differences in what each proxy penalizes. Depth is a coarse capacity control tied to worst-case decision path length (interpretability and latency). The \texttt{zlib} proxy is closer to an MDL-style description-length penalty and in this experiment pushes strongly toward simpler trees (here, favoring a shallow tree that underfits the noisy training data but improves clean-test SSE). The entropy proxy is smoother but can be numerically narrow (here spanning only a small range), so in finite regimes it may have limited granularity and therefore induces a different trade-off curve.

The key point is that the elbow point and, correspondingly, any susceptibility peak $\chi(\lambda)=\mathrm{Var}_{\pi_{\lambda}}[\mathrm{Comp}(S)]$ is not a universal invariant of the task alone: it is a statement about \emph{loss vs.\ the particular operational complexity} that the user chose to measure. Our framework is developed for exactly this reality: one specifies constraints, chooses $\mathrm{Comp}(S)$ to faithfully represent them, and then reads off the salient transitions from the resulting loss--complexity landscape.

\section{Conclusion}
We have constructed a complete theoretical and computational framework to provide a practical and computable equivalent of the Kolmogorov structure function, and established a new information–scattering analogy that predicts resonance phenomena. This leads to an efficient methods useful in model selection that requires only already existing Bayesian optimizers such as HyperOpt \cite{bergstra2013making} or Optuna \cite{akiba2019optuna} to run the necessary analysis. Our method finds the optimal goodness-of-fit vs model complexity tradeoff after which overfitting occurs. Experimental results validate our theoretical claims.

\section*{Acknowledgments} 
This material is based upon work supported by the Google~Cloud Research Award number GCP19980904. This work was also supported by the Wolfram Institute for Computational Foundations of Science, and by the John Templeton Foundation. 

\printbibliography

\end{document}